\documentclass[11pt,a4paper]{article}
\usepackage[latin1]{inputenc}
\usepackage[margin=3cm]{geometry}
\usepackage{amsmath}
\usepackage{amsfonts}
\usepackage{amssymb}
\usepackage{graphicx}
\usepackage{amsthm}
\usepackage{enumerate,hyperref} 
\usepackage[version=3]{mhchem}
\usepackage[all]{xy}
\usepackage{adjustbox}
\usepackage{arydshln} 
\usepackage{lineno,xfrac}
\usepackage{marginnote}
\usepackage{multirow}
\usepackage{pgf,tikz}
\usepackage{url}

\usepackage[sort&compress,comma,square,numbers]{natbib}

\theoremstyle{plain}
\newtheorem{thm}{Theorem}[section]
\newtheorem{proposition}[thm]{Proposition}

\newtheorem{lem}[thm]{Lemma}

\theoremstyle{definition}
\newtheorem{definition}[thm]{Definition}
\newtheorem{algorithm}[thm]{Algorithm}
\newtheorem{remark}[thm]{Remark}
\newtheorem{example}[thm]{Example}

\newcommand{\cN}{\mathcal{N}}
\newcommand{\cS}{\mathcal{S}}
\newcommand{\cC}{\mathcal{C}}
\newcommand{\cR}{\mathcal{R}}
\newcommand{\cY}{\mathcal{Y}}
\newcommand{\cX}{\mathcal{X}}

\newcommand{\R}{\mathbb{R}}

\newcommand{\Z}{{\mathbb Z}}

\renewcommand{\k}{\kappa}
\newcommand{\e}{n}

\DeclareMathOperator{\diag}{diag}

\DeclareMathOperator{\Rem}{rem}

\DeclareMathOperator{\rank}{rank}
\DeclareMathOperator{\Mult}{Mult}
\DeclareMathOperator{\Circuits}{Circuits}

 \newcommand\blue[1]{\textcolor{black}{#1}}
\newcommand\red[1]{\textcolor{black}{#1}}

\begin{document}

\author{AmirHosein Sadeghimanesh$^1$, Elisenda Feliu$^{1,2}$}
\title{The multistationarity structure of networks with intermediates and a binomial core network}
\date{\today}

\footnotetext[1]{Department of Mathematical Sciences, University of Copenhagen, Universitetsparken 5, 2100 Copenhagen, Denmark}
\footnotetext[2]{Corresponding author: efeliu@math.ku.dk}

	\maketitle
	\begin{abstract}
This work addresses whether a reaction network, taken with mass-action kinetics, is multistationary, that is, admits more than one positive steady state in some stoichiometric compatibility class. 
We build on previous work on the effect that removing or adding intermediates has on multistationarity,  and also on methods to detect multistationarity for networks with a binomial steady state ideal. In particular, we provide a new determinant criterion to decide whether a network is multistationary, which  applies when the network obtained by removing intermediates has a binomial steady state ideal. We apply this method to easily characterize which subsets of complexes are responsible for multistationarity; this is what we call the \emph{multistationarity structure} of the network. We use our approach to compute the multistationarity structure of the $n$-site sequential  distributive  phosphorylation cycle for arbitrary $n$.

	\end{abstract}
	
	\section*{Introduction}
Given a reaction network, an \emph{intermediate} is a species that does not interact with any other species, is produced by at least one reaction, and consumed by at least one reaction. Typical intermediates $Y$ arise in  Michaelis-Menten type mechanisms as
\[ c \ce{<=>} Y \ce{->} c',\]
where $c,c'$ are arbitrary complexes. Removal of the intermediates of  a network yields a new network, called the \emph{core network}, as introduced in \cite{Feliu-Simplifying} (and further generalized in \cite{Saez:reduction}). 
For example, removal of $Y$ from the mechanism above gives the reaction $c\rightarrow c'$. 

We consider mass-action kinetics, such that the system of ordinary differential equations (ODEs) modeling the evolution of the concentrations of the species  in time is polynomial. 
As shown in \cite[Theorem 5.1]{Feliu-Simplifying},  multistationarity of the core network implies multistationarity of the original (extended) network, provided  a technical realization condition is satisfied. Further, whether an extended network is multistationary  depends only on  the set of complexes of the core network that react to an intermediate. These complexes are called \emph{inputs}. 
The subsets of complexes that give rise to multistationarity define the \emph{multistationarity structure} of the core network. 

In this work we present an approach to find the multistationarity structure of core networks, and thereby provide a fast way to decide whether a given extended network is multistationary or not; indeed, it suffices to find the set of inputs of our network and check whether it belongs to the multistationarity structure. 

The method applies to core networks that are \emph{binomial} (the ideal generated by the steady state polynomials is binomial). For these networks, a method to decide on multistationarity was introduced in \cite{Dickenstein-Toric}, based on the computation of sign vectors. Under some extra assumptions, another method relying on the computation of a symbolic determinant and inspection of the sign of its coefficients is presented in  \cite{Feliu-Sign}  (see Theorem \ref{Theorem:Multistationarity_and_Determinant_Condition}).

The first main result of this paper is Theorem \ref{Theorem:Usage_of_Canonical_networks_for_multistationarity}, where we combine the results in \cite{Feliu-Simplifying,Feliu-Sign}  into a new determinant criterion for multistationarity. The criterion applies to extended networks with binomial core network, even though the original network might not be binomial. \blue{One of the crucial hypothesis of this theorem is the fulfillment of the technical realization condition. We  show in Proposition~\ref{Lemma:Realization_more_classes} that these conditions hold  for typical types of intermediates like the Michaelis-Menten mechanism above, bypassing thereby the need to perform costly computations and making Theorem \ref{Theorem:Usage_of_Canonical_networks_for_multistationarity} useful in practice. }

The second main result   is Theorem \ref{Proposition:Multistationarity_and_binomial_extensions}, which removes the technical realization condition in \cite{Feliu-Simplifying} for concluding that an extended network is multistationary provided the core network is. Instead, we require that both the core and the extended networks are binomial  (in a compatible way). 
\blue{This result is appealing since it might not be straightforward to verify that the realization conditions  are satisfied for networks not covered in our study in Section~\ref{sec:Realization_conditions}. }

The third main contribution  is Algorithm \ref{Algorith:Circuits_Main}, which returns the multistationarity structure of a binomial core network  based on the determinant criterion. The algorithm relies on the study of the signs of a polynomial obtained after the computation of the determinant of a symbolic matrix. Our approach is more direct than testing if the  extended network is multistationary for all subsets of input complexes. 
\blue{We apply our method to completely determine the multistationarity structure of a biologically relevant network, namely the $\e$-site  distributive  sequential phosphorylation cycle, for arbitrary $\e$. This results in our last main contribution, namely Theorem~\ref{Theorem:Circuits_of_n_site_phosphorylation}, which characterizes multistationarity for all variants of the $\e$-site phosphorylation cycle obtained by altering the configuration of intermediates. }
With this example we further illustrate how our results enable the study of a whole family of networks at once.

\blue{Finally, in this work we clarify and connect previous results on multistationarity, determinant conditions, and intermediates \cite{Dickenstein-Toric,Feliu-Sign,Feliu-Simplifying}. For example, this leads to the elaboration on the realization conditions in Section~\ref{sec:Realization_conditions},  and to highlighting Theorem~\ref{Theorem:Multistationarity_and_Determinant_Condition}, which is presented in a more general form in  \cite{Feliu-Sign}, but where the important application to binomial networks might not be appreciated. }

\medskip
  The structure of the paper is as follows. In \S\ref{sec:Reaction_networks} we introduce basic concepts on  reaction networks and multistationarity. In \S \ref{sec:Binomial_networks_and_multistationarity} we introduce (complete) binomial  networks and the determinant criterion for determining  multistationarity.  In \S\ref{sec:Intermediates_and_multistationarity} we  focus on intermediates and give the determinant criterion for multistationarity applicable to  extended networks with a binomial core network. In \S\ref{sec:Lifting_multistationarity_and_the_multistationarity_structure} we link multistationarity of the core and extended networks, and in particular study the multistationary structure. Finally, in \S \ref{sec:Realization_conditions} we expand on how to check the realization conditions and determine network structures that satisfy them.

	\medskip
	\textbf{Notation. }
	Subscripts $\geq 0$, $>0$ for $\R$ refer to the non-negative and positive real num\-bers. The sets $\{1,\dots,m\}$ and $\{m_1,\dots,m_2\}$ are respectively denoted by $[m]$ and $[m_1,m_2]$. In particular $[m]=[1,m]$.
The cardinality of a set $A$ is denoted by $|A|$.

Consider two vectors $u,v\in\R^n$.  The scalar product of $u$ and $v$ is denoted by $u\cdot v$.
The vector $v^{u}$ is defined as $\prod_{j=1}^n v_j^{u_j}$, and for a matrix $M\in\R^{n\times m}$ with column vectors $u^{(i)}$, $i\in[m]$, the vector $v^M$ is the vector whose $i$-th entry is $v^{u^{(i)}}$. 
	We let $\diag(v)$ be the diagonal matrix with diagonal $v$ and 
\begin{equation}\label{eq:Mv} M_v=M\diag(v).
\end{equation}
The sign vector of $v$, $\sigma(v)\in\{-1,0,1\}^n$, is defined for $ i\in[n]$ as
\begin{equation}\label{eq:sigma}\sigma(v)_i=\left\{\begin{array}{lll}
	1 & \text{if} & v_i>0\\
	0 & \text{if} & v_i=0\\
	-1 & \text{if} & v_i<0.
	\end{array}\right.
\end{equation}

	\section{Reaction networks}\label{sec:Reaction_networks}
In this section we briefly introduce the ingredients from chemical reaction network theory needed in the sequel. See for example \cite{gunawardena-notes,feinbergnotes}.
A reaction network, also called a \emph{network},  is a triplet of finite sets $\cN=(\cS,\cC,\cR)$. 
The three sets are called respectively the set of \emph{species}, \emph{complexes} and \emph{reactions}.
The elements of $\cC$ are finite linear combinations of the species with non-negative integer coefficients. The set of reactions consists of ordered pairs $(c,c')$ of complexes, denoted $c\rightarrow c'$. \blue{We let $r$ be the cardinality of $\cR$. }

After fixing an order on $\cS$, write $\cS=\{X_1,\dots,X_n\}$. We identify a complex $c\in\cC$ with  the vector in $\R^n$ whose $i$-th entry is the coefficient of $X_i$ in $c$.  Therefore a complex $c$ is either given as  $\sum_{X\in\cS}c_XX$ or by the corresponding vector
 (again denoted $c$). 

There is a natural digraph associated with a network, with vertex set $\cC$ and edge set $\cR$. We often identify the reaction network with the digraph for simplicity. 
	The stoichiometric matrix $N$ is an \blue{$n \times r$} matrix whose column vectors are $c'-c$ for each reaction $c\rightarrow c'\in\cR$. This matrix depends on a fixed order of the set of reactions. The (real) column space of $N$ is called the \emph{stoichiometric subspace} and is denoted by $S$. Its dimension, that is, the rank of $N$,  is  the \emph{rank} of the network.

In this work we consider so-called \emph{mass-action kinetics}. Under this assumption, the evolution of the concentration of the species in time is modeled by means of   a polynomial ODE system as follows.  
First, a positive real number $k_{c\rightarrow c'}$ is assigned to each reaction $c\rightarrow c'$. This number is called the \emph{reaction rate constant} and often written as a label of the reaction in the associated digraph.  We interchangeably write $k_i=k_{c\rightarrow c'}$ if $c\rightarrow c'$ is the $i$-th reaction and write the vector of reaction rate constants   $k\in \R_{>0}^{r}$. 

Next, we let $x=(x_1,\dots,x_n)$ denote the vector of the concentrations of $X_1,\dots,X_n$; note that in examples we simply use corresponding  lower-case letters to denote concentrations. 
Given $x\in \R^{n}$, we define the vector $\psi(x)\in  \R^{r}$ as 
\[\psi(x)_i=x^c,\qquad  \textrm{if }c\rightarrow c' \textrm{ is the $i$-th reaction}.\]
Now  the ODE system associated with  the network and $k\in \R^{r}_{>0}$ is
\begin{equation}\label{eq:ODE}
\frac{dx}{dt}=F_{k}(x), \qquad \textrm{where } F_{k}(x)=N_k\psi(x),\quad x\in \R^{n}_{\geq 0}.
\end{equation}
Recall that $N_k=N\diag(k)$, c.f.~\eqref{eq:Mv}.

	The solution to \eqref{eq:ODE} with an initial condition $x_0\in \R^n_{\geq 0}$ is confined to  the \emph{stoichiometric compatibility class} of $x_0$:  $(x_0 + S)\cap \R^n_{\geq 0}$ \cite[Remark 3.4]{FeinbergExistenceAndUniquness}. 
Equations for these classes are found as follows. Let $d$ be the \emph{corank} of the network, that is, $d=n-\rank(N)$.
	A matrix $Z\in \R^{d\times n}$ whose rows form a basis of the orthogonal complement  $S^\perp$ of $S$ is called a \emph{matrix of conservation laws}. Then the set $(x_0 + S)\cap \R^n_{\geq 0}$ agrees with the set
	\[\{u\in\R_{\geq 0}^n \mid Z u=Z x_0\}.\]

	A positive \emph{steady state} is a solution to the system $F_k(x)=0$ in $\R^{n}_{>0}$. Since we would like to treat the values of $k$ as unknown, we view the polynomials $F_{k,i}(x)$ as polynomials in  the ring $\mathbb{R}(k)[x]$ by regarding $k$ as parameters  instead of positive real numbers. When we do this, we write $F_i(x)$.
Then $F_1(x),\dots,F_n(x)$ are called the \emph{steady state polynomials} of the network, and the ideal $I$ they generate is  the \emph{steady state ideal}:
\[ I= \big\langle F_1(x),\dots,F_n(x)\big\rangle \subseteq \mathbb{R}(k)[x].\]
Throughout this work, given an element or subset $B$ of  $\mathbb{R}(k)[x]$, we  denote by $B_k$ the specialization of $B$ to a given value of $k$. 
	
	It is always possible to find a basis (a set of generators) of the steady state ideal of a network with cardinality equal to the rank of the network.  Indeed, $d$ of the steady state polynomials are redundant since they can be expressed as a linear combination of the $n-d$ remaining polynomials.
	
	\begin{definition}\label{Definition:Multistationarity}
		We say that a reaction network   is \emph{multistationary} if there exists a strictly positive vector $k\in\R^{r}_{>0}$ such that the system $F_{k,i}(x)=0$, $i\in[n]$, has more than one positive solution in a stoichiometric compatibility class. Alternatively, given a matrix of conservation laws $Z$, the system 
\[N_k\psi(x)=0\quad \text{ and }\quad Zx=\alpha\] 
has at least two positive solutions for some positive $k$ and $\alpha\in \R^d$.
	\end{definition}

	\section{Binomial networks and multistationarity}\label{sec:Binomial_networks_and_multistationarity}
	In this section we discuss and expand known results on determining whether a network is multistationary when the steady state ideal is binomial. The main references are \cite{Dickenstein-Toric,Feliu-Sign}. 
	An ideal is binomial if it admits a binomial basis, that is, a basis with all polynomials having at most two terms. It is   well known  that an ideal is binomial if and only if the reduced Gr\"obner basis in an arbitrary  monomial order consists only of  binomials \cite[Corollary 1.2]{Eisenbud-Binomial}.

We start with an observation on changing bases in $\mathbb{R}(k)[x]$ for $k=(k_1,\dots,k_r)$ and $x=(x_1,\dots,x_n)$.
	For a set of polynomials $A$ in a polynomial ring $K[x]$, we let $V(A)\subseteq K^n$ denote its solution set, which agrees with $V(\langle A \rangle)$.
	If $B$ and $B'$ are bases of the same ideal $I$ in $\mathbb{R}(k)[x]$, then $V(B)=V(B')\subseteq\big(\mathbb{R}(k)\big)^n$. However, this does not imply that the specializations to real values $k\in\mathbb{R}^r$ agree, that is, it can happen that  $V(B_{k})\neq V(B'_{k})\subseteq\mathbb{R}^n$. Since we want to study the steady state ideal in $\mathbb{R}(k)[x]$ but obtain results for specific values of $k$, we introduce the following definition.

	\begin{definition}\label{Definition:Binomial_networks}
		Let $\cN$ be a reaction network and $B\subseteq\R(k)[x]$ the set of steady state polynomials. A basis $B'$ of the steady state ideal of $\cN$ is called \emph{admissible} if for every $k\in\R_{>0}^{r}$ it holds that \blue{$V(B_k)\cap \R^n_{\geq 0}=V(B_k')\cap \R^n_{\geq 0}$}. The network $\cN$ is a \emph{binomial network} if the steady state ideal has an admissible binomial basis.
	\end{definition}

	We consider a sufficient condition to decide whether the solution sets of two parametric systems agree for any specialization of the parameters, and in particular, for when a basis of the steady state ideal  is admissible. 	Let $B=\{f_1,\dots,f_\ell\}$ and $B'=\{f_1',\dots,f'_{\ell'}\}$. We consider representations of $B$ in terms of $B'$ and \emph{vice versa}, that is, we write
\begin{equation}\label{eq:representation}
f_i=\sum_{j\in [\ell']}\tfrac{h_{ij}}{h_i}f_j',\quad \textrm{for }i\in[\ell],\qquad\textrm{and}\qquad  f_i'=\sum_{j\in [\ell]}\tfrac{h_{ij}'}{h_i'}f_j, \quad\textrm{for }i\in[\ell'],
\end{equation}
with 
 $h_i,h_i'\in\mathbb{R}[k]$ and $h_{i1},\dots,h_{i\ell'},h_{i1}',\dots,h_{i\ell}'\in\mathbb{R}[k][x]$. Note that these representations might not be unique.

	\begin{lem}\label{Lemma:Bases_having_same_solution_sets}
With the notation above, given two bases $B$ and $B'$ of an ideal in $\mathbb{R}(k)[x]$, if $k^\star$ is not in the zero set of $(\prod_{i=1}^\ell h_i)(\prod_{i=1}^{\ell'} h_i')$, then  $\langle B_{k^\star}\rangle=\langle B'_{k^\star}\rangle$.
	\end{lem}
	\begin{proof}
For all $i\in[\ell]$ and $j\in[\ell']$ we have that  $\tfrac{h_{k^\star,ij}}{h_{k^\star,i}},\tfrac{h'_{k^{\star},ji}}{h'_{k^\star,j}}\in\mathbb{R}[x]$ and the equalities in \eqref{eq:representation} specialize to $k^\star$.
Hence $B_{k^\star}\subseteq\langle B_{k^\star}'\rangle$ and  $B'_{k^\star}\subseteq\langle B_{k^\star}\rangle$ and so $\langle B_{k^\star}\rangle=\langle B'_{k^\star}\rangle$. 
	\end{proof}

In particular, if $(\prod_{i=1}^\ell h_i)(\prod_{i=1}^{\ell'} h_i')$ has no positive solution,  then $V(B_k)=V(B_k')$ for all positive $k$.

	\begin{example}\label{Example Surjecctivity dependence}
		Consider the following reaction network
\[ 
X_1 \ce{->[k_1]} 2X_1 \ce{<-[k_3]} X_2 \qquad X_1 \ce{->[k_2]} 2X_2 \ce{<-[k_4]} X_2.
\]
		The set of steady state polynomials is
		\[B=\big\{(k_1-k_2)x_1+(2k_3)x_2,(2k_2)x_1+(-k_3+k_4)x_2\big\},\]
and the set 		$B'=\{x_1-x_2,x_1-2x_2\}$
		is another basis of the steady state ideal in $\R(k)[x]$.
\blue{To see this, we note that
		\begin{align*}
		(k_1-k_2)x_1+(2k_3)x_2 &= (2k_1-2k_2+2k_3)(x_1-x_2)+(-k_1+k_2-2k_3)(x_1-2x_2)\\
		(2k_2)x_1+(-k_3+k_4)x_2 &= (4k_2-k_3+k_4)(x_1-x_2)+(-2k_2+k_3-k_4)(x_1-2x_2),
		\end{align*}
which gives that $B\subseteq\langle B'\rangle$. Similarly, we write
\begin{equation} \label{eq:ex1}
\begin{aligned}
x_1-x_2 &= \tfrac{(-2k_2+k_3-k_4) ((k_1-k_2)x_1+2k_3\, x_2)}{k_1k_3-k_1k_4+3k_2k_3+k_2k_4}+\tfrac{2k_2\, x_1+(-k_3+k_4)x_2}{k_1k_3-k_1k_4+3k_2k_3+k_2k_4}  \\
x_1-2x_2 &= \tfrac{(k_1-k_2)x_1+2k_3\, x_2}{k_1k_3-k_1k_4+3k_2k_3+k_2k_4}+\tfrac{2k_2\, x_1+(-k_3+k_4)x_2}{k_1k_3-k_1k_4+3k_2k_3+k_2k_4}.
\end{aligned}
\end{equation}
Therefore $B'\subseteq\langle B\rangle$ in $\R(k)[x_1,x_2]$, showing that both sets are bases of the same ideal. 
However, we have that $(a,b)\in \mathbb{R}_{>0}^2$ belongs to $V(B_k)$ with $k=(2\tfrac{b}{a}+1,1,1,2\tfrac{a}{b}+1)$. But $V(B'_k)\cap\R_{>0}^2=\emptyset$  for every choice of $k$.
		Hence $B'$ is not an admissible basis. Note that the denominators in \eqref{eq:ex1} vanish for this choice of $k$, and hence the assumptions of Lemma~\ref{Lemma:Bases_having_same_solution_sets} do not hold.}
	\end{example}

\color{black}
\begin{remark}\label{rk:binomial}
By computing a Gr\"obner basis of the steady state ideal and afterwards  verifying whether Lemma~\ref{Lemma:Bases_having_same_solution_sets} applies, one can determine whether a network is binomial. However, this approach will easily fail even for relatively small networks, due to the computational cost. 
Fortunately, the structure of realistic networks allows us often to conclude the existence of an admissible binomial basis without the need of finding a Gr\"obner basis.  Two strategies can be employed. The first consists of identifying \emph{intermediates} and using the results in \cite{Grobner}. This approach is explained in Remark~\ref{remark:algind} below. The second strategy exploits the frequent linearity in the steady state equations; indeed, admissible binomial bases are often found by simple row operations on $N$ or by performing linear combinations of the steady state polynomials. These can be detected by visual inspection or by Gaussian elimination. See \cite{Dickenstein-Structured} for an elaboration on this special case. Additionally, in \cite{Dickenstein-MESSI}, the authors introduce structural conditions on a class of networks called MESSI systems,  that guarantee that the steady state ideal is binomial, thereby bypassing the use of expensive computations. These systems are called \emph{s-toric}. 

A recommended initial check is to assign random values to the parameters and compute a Gr\"obner basis for this specialization. If the ideal is not binomial for this random choice, then the ideal cannot admit an admissible binomial basis. 
\end{remark}

\color{black}

The connection between binomial ideals and multistationarity is as follows. 
	 Consider a system of binomial equations in $\mathbb{R}(k)[x]$, say
\begin{equation}\label{eq:binomial_equations}
	p_1(k)x^{c_1}-p_1'(k)x^{c_1'}=0,\quad
	\dots \quad
	p_s(k)x^{c_s}-p_s'(k)x^{c_s'}=0	.
\end{equation}
If one of the equations has only one term, or the two terms of a binomial have the same sign, then the system does not admit positive solutions. If $p_i\neq 0$ and $p_i' \neq 0$ in $\R(k)$, the positive solutions to \eqref{eq:binomial_equations} are the positive solutions of the following system
	\[
	x^{c_1-c_1'}=\tfrac{p_1'(k)}{p_1(k)},\quad
	\dots\quad
	x^{c_s-c_s'}=\tfrac{p_s'(k)}{p_s(k)}
	.\]
\blue{The right-hand side of these expressions specialize at least to all values $k^\star\in \R^r_{>0}$ for which the denominators $p_i(k)$ do not vanish. }
	Letting
	\begin{equation}\label{Equation Gamma and M}
	\gamma(k):=\begin{bmatrix}
	\frac{p_1'(k)}{p_1(k)}\\
	\vdots\\
	\frac{p_s'(k)}{p_s(k)}
	\end{bmatrix}
	\qquad \text{ and } \qquad M :=\begin{bmatrix}
	c_1-c_1'\\
	\vdots\\
	c_s-c_s'
	\end{bmatrix}^T \in\R^{n\times s},
	\end{equation}
	the set of positive solutions of \eqref{eq:binomial_equations} for a positive vector $k$ such that $p_i(k)\neq 0$ for all $i\in [s]$ is
	\begin{equation}\label{Equation Binomial parametrization}
	\big\{x\in\mathbb{R}_{>0}^n\mid x^M=\gamma(k) \big\}.
	\end{equation}
 Let $M'$ be a matrix whose rows form a basis of the orthogonal complement of the row space of $M$. The solution set of $x^M=\gamma(k)$ is non-empty if and only if $M'\ln\hspace{-2pt}\big(\gamma(k)\big)=0$, where $\ln\hspace{-2pt}\big(\gamma(k)\big)$ is defined component-wise. To see why, take the logarithm of both sides, which gives $M^T \ln(x)=\ln\hspace{-2pt}\big(\gamma(k)\big)$, and impose that $\ln\hspace{-2pt}\big(\gamma(k)\big)$ belongs to the image \blue{of the transpose $M^T$ of $M$.}
	
	The parametrization of positive solutions of a binomial system \eqref{eq:binomial_equations}  in \eqref{Equation Binomial parametrization}, makes it possible to use results of \cite{Dickenstein-Toric,Feliu-Sign} for detecting multistationarity of binomial networks. \blue{Recall the sign vector $\sigma(\cdot)$ defined in \eqref{eq:sigma}.}
	
	\begin{thm}[\cite{Feliu-Sign}, Proposition 3.9 and Corollary 3.11]\label{Theorem Sign Condition}
		Let $\cN$ be a binomial network. Let $M$ and $\gamma(k)$ be as in (\ref{Equation Gamma and M}), obtained from an admissible binomial basis $B$ of $I$, and $Z$ be a matrix of conservation laws. Consider the following conditions:
		
\smallskip
\noindent		\emph{(surj) Surjectivity condition}:
		for all $x\in\mathbb{R}^{n}_{>0}$ there exists $k\in\mathbb{R}^r_{>0}$ such that $x^M=\gamma(k)$ (equivalently, $x\in V(B_k)$).

\smallskip		
\noindent		\emph{(sign) Sign Condition}:
		there exist $u,v\in\mathbb{R}^n\setminus\{0\}$ such that $M^T u=Zv=0$ and $\sigma(u)=\sigma(v)$.

\smallskip		
Then we have:
		\begin{itemize}
			\item[(i)] Assume (surj) is satisfied. Then $\cN$ is multistationary if and only if (sign) holds.
			\item[(ii)] If (sign) does not hold, then $\cN$ is not multistationary.
		\end{itemize}
	\end{thm}

	\begin{lem}[\cite{Dickenstein-Toric}]\label{Lemma Surjectivity Complete}
		A binomial network $\cN$ with stoichiometric matrix $N$ satisfies (surj) if and only if $\ker N\cap\mathbb{R}_{>0}^{r}\neq\emptyset$.
	\end{lem}
	\begin{proof}
The solution set of an admissible binomial basis agrees with the solution set of the set of steady state polynomials for every $k\in\R_{>0}^{r}$. Therefore (surj) is equivalent to the statement
		\[\text{for all }x\in\mathbb{R}_{>0}^{n}\text{, there exists }k\in\mathbb{R}_{>0}^{r}\text{ such that }N_k\psi(x)=0.\]
		Now $N_k\psi(x)=0$ is equivalent to $\diag(k)\psi(x)\in\ker N\cap \mathbb{R}_{>0}^{r}$. Thus  if $\ker N\cap\mathbb{R}_{>0}^{r}=\emptyset$, then (surj) fails. Conversely given $v\in\ker N\cap\mathbb{R}_{>0}^{r}$ and any $x\in\mathbb{R}_{>0}^{n}$, by taking $k=\tfrac{v}{\psi(x)}$ (defined component-wise), we see that (surj) holds.
	\end{proof}

Networks fulfilling the condition of Lemma~\ref{Lemma Surjectivity Complete} above are often called \emph{consistent}. 
Using Lemma \ref{Lemma Surjectivity Complete} we can check  (surj) algorithmically. For a matrix $N$, $U=\ker(N)\cap\R_{\geq 0}^{r}$ is a convex set. A set of vectors in $\R_{\geq 0}^{r}$ is an  \emph{extremal generating set} for $U$ if their non-negative linear combinations generate $U$ and none of them is a non-negative combination of the rest. Then $U$ contains a strictly positive vector   if and only if the sum of the vectors in an extremal generating set is positive. 
To find an extremal generating set for a convex set one can use  existing algorithms, e.g.~\cite[Appendix B]{ExtremePathways}.

	Let $\cN$ be a binomial reaction network with an admissible binomial basis $B$. Let $M$ be the exponent matrix in the parametrization of its positive steady states as in \eqref{Equation Binomial parametrization}, and $Z$ be a matrix of conservation laws. For $\lambda=(\lambda_1,\dots,\lambda_n)$ a vector of indeterminates, define 
\begin{equation}\label{eq:Gamma}
\Gamma=\begin{bmatrix}
	\blue{(M^T)_\lambda} \, \\ Z
	\end{bmatrix} \quad  \in\  (\R[\lambda])^{n\times n}.
\end{equation}
\blue{Recall the definition of $(M^T)_\lambda$ in \eqref{eq:Mv}.}
Consider the following conditions:
	
\smallskip
\noindent
	\emph{(rank) Rank Condition \blue{on $B$}}:  \blue{The number of elements of $B$, equivalently  the number of columns of the exponent matrix $M$ derived from $B$, is equal to the rank of the network.}

\smallskip
\noindent	
	\emph{(det) Determinant Condition}: 
\blue{For a complete binomial network, pick any admissible binomial basis that satisfies (surj) and (rank), a matrix $Z$ of conservation laws, and consider the corresponding matrix $\Gamma$ in \eqref{eq:Gamma}. Then viewed as a polynomial in $\lambda$, $\det(\Gamma)$ is either zero or  has at least one positive and at least one negative coefficient.
	}

Note that by  \cite[Lemma 2.11]{Feliu-Sign}, $\det(\Gamma)$ is a polynomial in $\lambda$ that is linear or constant in each $\lambda_i$. 	Theorem 2.13 of \cite{Feliu-Sign} states that provided (rank) is fulfilled, then (sign) holds if and only if (det) holds. \blue{In particular (det) is well formulated, since it does not depend on the choice of basis. } Combining this with Theorem \ref{Theorem Sign Condition}, $\cN$ is multistationary if and only if (det) holds. This \blue{yields} the following definition and theorem.
	
	\begin{definition}\label{Definition:Complete_binomial_networks}
		A binomial network $\cN$ is \emph{complete} if (surj) and (rank) hold for an admissible binomial basis.
	\end{definition}
	
	\begin{thm}[Determinant criterion for complete binomial networks]\label{Theorem:Multistationarity_and_Determinant_Condition}
\blue{Consider a complete binomial network. 
Then the network is multistationary  if and only if (det) holds. }
	\end{thm}

	\begin{example}\label{Example:A_Biological_Circuit_Determinant}
		Consider the following network modeling a simple biological circuit:
        \[X_1\ce{<-[\k_1]}0\ce{<-[\k_7]}X_2\qquad X_1+E\ce{<=>[\k_2][\k_3]}Y_1\ce{->[\k_4]}X_2+E\qquad 2X_1+E\ce{<=>[\k_5][\k_6]}Y_2.\] 
\blue{The steady state polynomials for $Y_1$ and $Y_2$ are binomial. By using these polynomials equated to zero to write $y_1$ and $y_2$ in terms of $x_1,x_2,e$ and substituting in the steady state polynomials of $X_1$ and $X_2$, we obtain the following  admissible binomial basis of the steady state ideal:}
		\[B=\Big\{\k_1-\tfrac{\k_2\k_4}{\k_3+\k_4}x_1e,\,\tfrac{\k_2\k_4}{\k_3+\k_4}x_1e-\k_7x_2,\,\k_2x_1e-(\k_3+\k_4)y_1,\,\k_5x_1^2e-\k_6y_2\Big\}.\]
		The  matrices $M$, $Z$ and $N$ are
{\small		\[M=
		\left[\begin{array}{rrrr}
		1 & 1 & -1 & -2\\
		0 & -1 & 0 & 0\\
		1 & 1 & -1 & -1\\
		0 & 0 & 1 & 0\\
		0 & 0 & 0 & 1
		\end{array}\right],\,Z=\begin{bmatrix}
		0 & 0 & 1 & 1 & 1
		\end{bmatrix},\,N=\left[\begin{array}{rrrrrrr}
		1 & -1 & 1 & 0 & -2 & 2 & 0\\
		0 & 0 & 0 & 1 & 0 & 0 & -1\\
		0 & -1 & 1 & 1 & -1 & 1 & 0\\
		0 & 1 & -1 & -1 & 0 & 0 & 0\\
		0 & 0 & 0 & 0 & 1 & -1 & 0
		\end{array}\right].\]}The set $P=\{(1,1,0,1,0,0,1),(0,1,1,0,0,0,0),(0,0,0,0,1,1,0)\}$ is an extremal generating set for  $\ker(N)\cap\R_{\geq 0}^7$. 		Since the sum of the vectors in $P$ has all entries positive,  (surj) holds by Lemma \ref{Lemma Surjectivity Complete}.
		Now taking $u=v=(-1,0,1,0,-1)$, the condition (sign) holds. Therefore this network is multistationary by Theorem \ref{Theorem Sign Condition}. 
	
Alternatively, we have  $\rank(N)=4$ and $B$ has 4 elements; hence (rank) holds and this binomial network is complete. We have
		\[\det(\Gamma)=\left|\begin{array}{ccccc}
		\lambda_1 & 0 & \lambda_3 & 0 & 0\\
		\lambda_1 & -\lambda_2 & \lambda_3 & 0 & 0\\
		-\lambda_1 & 0 & -\lambda_3 & \lambda_4 & 0\\
		-2\lambda_1 & 0 & -\lambda_3 & 0 & \lambda_5\\
		\hdashline[2pt/2pt]
		0 & 0 & 1 & 1 & 1
		\end{array}\right|=\lambda_1\lambda_2\lambda_3\lambda_4-\lambda_1\lambda_2\lambda_4\lambda_5.\]
The network is multistationary  by Theorem \ref{Theorem:Multistationarity_and_Determinant_Condition} since  (det) holds.
	\end{example}

\blue{
\begin{remark}
A determinant condition for multistationarity appears in Lemma 4.4 and Theorem 4.6(A) in \cite{Dickenstein-Structured}, which applies to a special type of binomial networks. The polynomial $B(x)$ in \cite{Dickenstein-Structured} is up to sign the polynomial $\det(\Gamma)$ here. So the results  in \cite{Dickenstein-Structured} actually hold with full generality for complete  binomial networks, as was already shown in \cite{Feliu-Sign}, and as recalled here  in Theorem~\ref{Theorem:Multistationarity_and_Determinant_Condition}. The conditions imposed on binomial networks in  \cite{Dickenstein-Structured} imply that the network is binomial and satisfies (rank) for an admissible basis (c.f. \cite[Remark 3.5]{Dickenstein-Structured}). 
\end{remark}
}

	\section{Intermediates and multistationarity}\label{sec:Intermediates_and_multistationarity}
In this section we introduce a particular type of species, \emph{intermediates}, and extended and core networks obtained by adding or removing intermediates. We proceed to present results on multistationarity of extended and core networks from \cite{Feliu-Simplifying} and specifically for binomial networks from \cite{Grobner}.

\subsection{Intermediates}

	A species $Y$ is an \emph{intermediate} if it is also a complex, that is belongs to $\cC$, only appears in the complex $Y$, and further both  the outdegree and indegree of $Y$ are at least one in the digraph of the network \cite{Feliu-Simplifying}. Given a set of intermediates $\cY=\{Y_1,\dots,Y_m\}$, let  $\cX=\cS\setminus \cY=\{X_1,\dots,X_n\}$ be the set of  non-intermediates. Then $\cS$ is the disjoint union of $\cX$ and $\cY$.  From now on, the species are ordered such that intermediates are after non-intermediates. Then by $(x,y)$ we mean the vector $(x_1,\dots,x_n,y_1,\dots,y_m)$.
A complex that is not an intermediate is called a non-intermediate complex.

Given an intermediate $Y$, an \emph{input} for $Y$ is a non-intermediate complex $c$ such that there exists a directed reaction path from $c$ to $Y$ with all vertices other than $c$ being intermediates. The intermediate $Y$ is an $\ell$-input intermediate if it has $\ell$ inputs \cite{Grobner}.
	
	\begin{definition}[\cite{Feliu-Simplifying}]\label{Definition:Notation_Extension_via_intermediates} 
		Let $\cN=(\cS,\cC,\cR)$ and $\widetilde{\cN}=(\widetilde{\cS},\widetilde{\cC},\widetilde{\cR})$ be two reaction networks. We say that $\widetilde{\cN}$ is an \emph{extension} of $\cN$ via the addition of   intermediates $Y_1,\dots,Y_m$ if
		\begin{enumerate}[(i)]
			\item $\cY=\{Y_1,\dots,Y_m\}$ is a set of intermediates of $\widetilde{\cN}$.
			\item $\cS\cup \cY = \widetilde{\cS}$ and  $\cC\cup \cY = \widetilde{\cC}$.
			\item $c\rightarrow c'\in\cR$ if and only if there is a directed path from $c$ to $c'$ in the  digraph associated with $\widetilde{\cN}$, such that all vertices other than $c$ and $c'$ belong to $\cY$ (there might be none).
		\end{enumerate}
		In this case $\cN$ is called the \emph{core network} of $\widetilde{\cN}$.
	\end{definition}

\blue{We use $\kappa$ to denote the vector of reaction rate constants of the extended network and, in general, symbols with a tilde $\sim$ refer to the extended network. }
		Let $\cN$ be a reaction network and $\widetilde{\cN}$ an extension of it via the addition of $m$ intermediates, $Y_1,\dots,Y_m$. Choose an input complex $c_i$ for each intermediate $Y_i$ and let $\begin{bmatrix}
	c_1 & \dots & c_m
	\end{bmatrix}\in\R^{n\times m}$ be the matrix whose columns are $c_1,\dots,c_m$.
 It follows from Theorem 2.1 of \cite{Feliu-Simplifying} that  if $Z$ is a matrix of conservation laws for $\cN$, then a matrix of conservation laws for $\widetilde{\cN}$ is
	\begin{equation}\label{Eq:Matrix_of_Conservation_Laws_for_extended-networks}
	\widetilde{Z}=\left[\begin{array}{c;{2pt/2pt}c}
	Z & Z\begin{bmatrix}
	c_1 & \dots & c_m
	\end{bmatrix}
	\end{array}\right]\in\R^{d\times(n+m)}.
	\end{equation}
	In particular the corank of $\cN$ and $\widetilde{\cN}$ agree and the rank of $\widetilde{\cN}$ is the rank of $\cN$ plus $m$.

We next introduce a simple type of extended networks via the addition of intermediates that are useful in the study of multistationarity, see Theorem~\ref{Theorem:Usage_of_Canonical_networks_for_multistationarity}. 

	\begin{definition}\label{Definition Canonical Networks}
		Let $\cN$ be a network and $C=\{c_1,\dots,c_m\}\subseteq\cC$.
 The \emph{canonical extension} of $\cN$ associated with $C$,  denoted by $\widetilde{\cN}_C=(\widetilde{\cS}_C,\widetilde{\cC}_C,\widetilde{\cR}_C)$, is the extension of $\cN$ via the addition of 1-input intermediates $Y_1,\dots,Y_m$  such that
		\[\widetilde{\cR}_C=\cR\cup \big\{c_i\rightleftharpoons Y_i \mid i\in [m]\big\}.\]
The canonical extension associated with $C=\cC$ is called the \emph{largest canonical extension}. 
	\end{definition}

	We now review the key results in  \cite{Feliu-Simplifying} that relate the steady states of extended and core networks.
We start by studying the steady state polynomials  of the two networks.
	Let $\widetilde{\cN}$ be an extension of $\cN$ via the addition of intermediates $Y_1,\dots,Y_m$ and $C\subseteq\cC$ be the set of input complexes. The steady state polynomials associated with the intermediates yield a system
	$F_{n+1}(x,y)=\dots=F_{n+m}(x,y)=0$
	that is linear in $y$ and square, \blue{that is, viewed as a system in $y_1,\dots,y_m$, it is linear and has $m$ equations and $m$ variables}. As shown in \cite{Feliu-Simplifying}, the solution to this linear system is of the form 
\begin{equation}\label{eq:defmu}
y_i=\sum_{c\in C}\mu_{i,c}x^c, \quad i\in[m]
\end{equation}
	where $\mu_{i,c}$ is a rational function in $\R(\k)$ with all non-zero coefficients positive (see also \eqref{eq:muic}).
\blue{Recall that using the fixed orders on $\cR$ and on $\widetilde{\cR}$ respectively, we identify   
$$k=(k_1,\dots,k_r)=(k_{c\rightarrow c'})_{c\rightarrow c'\in \cR}\quad\textrm{and}\quad \kappa=(\kappa_1,\dots,\kappa_{\widetilde{r}})= (\kappa_{\widetilde{c} \rightarrow \widetilde{c}'})_{\widetilde{c} \rightarrow \widetilde{c}' \in \widetilde{\cR} }.$$
}
Consider the following map 
	\begin{equation}\label{Equation:phi}
	\begin{array}{lrll}
	\phi \colon & \mathbb{R}[k] & \longrightarrow & \mathbb{R}(\kappa)\\
	& k_{c\rightarrow c'} & \longmapsto & \blue{\phi_{c\rightarrow c'}}=  \kappa_{c\rightarrow c'}+\sum_{i=1}^m\kappa_{Y_i\rightarrow c'}\, \mu_{i,c},\end{array}
	\end{equation}
	where it is understood that $\kappa_{c\rightarrow c'} =0$ and $\kappa_{Y_i\rightarrow c'}=0$ if $c\rightarrow c'$ and $Y_i\rightarrow c'$ do not belong to $\widetilde{\cR}$  respectively. 
\blue{That is, the image of $k_{c\rightarrow c'}$ is a rational function $\phi_{c\rightarrow c'}$ in $\kappa$, which specializes to all positive vectors $\kappa$. Given 
a vector $\kappa\in \R^{\widetilde{r}}_{>0}$, we let $\phi^*(\kappa)\in \R^r_{>0}$ be defined component-wise  by evaluating $\phi_{c\rightarrow c'}$ at $\kappa$, that is $\phi^*(\kappa)_{c\rightarrow c'}= \phi_{c\rightarrow c'}(\kappa)$ for all $c\rightarrow c'\in \mathcal{R}$.}
Then the steady state polynomials  $F,\widetilde{F}$ of $\cN$ and $\widetilde{\cN}$ for non-intermediate species relate in the following way  (see \cite{Feliu-Simplifying}):
\[\widetilde{F}_{\k,i}\Big(x,\sum\nolimits_{c\in C}\mu_{1,c} x^c,\dots,\sum\nolimits_{c\in C}\mu_{m,c} x^c\Big)=F_{\blue{\phi^*(\k)},i}(x),\qquad i\in[n].\]

Given $\sfrac{f}{g}\in \R(k)$ such that $\phi(g)\neq 0$, then $\phi(\sfrac{f}{g})$  is well defined in $\mathbb{R}(\kappa)$. If \blue{$G\in \R(k)[x]$} is a polynomial in $x$ such that all coefficients are rational functions with  non-vanishing denominator upon applying $\phi$, then we consider the polynomial $\Phi(G)$ in $\R(\k)[x,y]$ obtained by applying $\phi$ on the coefficients of $G$. 
In particular, if the rational functions $\phi_{c\rightarrow c'}$ are algebraically independent over $\R$, then there is no polynomial with coefficients in $\R$ that identically vanishes when evaluated on the image of $\phi$. Then the map $\phi$ extends to a map of polynomial rings 
$\Phi\colon \mathbb{R}(k)[x]\rightarrow\mathbb{R}(\kappa)[x,y].$
Strategies to check this algebraic independence condition as well as classes of intermediates that satisfy it  are described in  \cite[\S 4]{Grobner}.   In particular the rational functions $\phi_{c\rightarrow c'}$ are algebraically independent over $\R$ for all canonical extensions by \cite[Corollary 4.6]{Grobner}.
	
\medskip
In order to introduce Theorem~\ref{Theorem Multistationartity and Realization} below, we need to consider the following conditions. 
Let $\omega_1,\dots,\omega_d$ be a basis of $S^\perp$ and 
$C'\subseteq C$ consist of the complexes $c$ such that $\omega_j\cdot c\neq 0$ for some $j\in[d]$ (i.e.~$c\notin S$).
We define two realization conditions on the reaction rate constants of $\cN$ and   $\widetilde{\cN}$:
	\begin{enumerate}[(i)]
		\item 	\emph{Realization condition:}\\
		For all $k\in\R^{r}_{>0}$, there exists $\kappa\in\R^{\widetilde{r}}_{>0}$ such that \blue{$k=\phi^*(\k)$.}
		\item \emph{Generalized realization condition: }\\
		For all $k\in\R^{r}_{>0}$ and $\alpha\in\R^{C'}_{>0}$, there exists $\kappa\in\mathbb{R}^{\widetilde{r}}_{>0}$ such that 
		\[ \blue{k=\phi^*(\kappa)} \qquad \textrm{and}\qquad 	      \alpha_{c}=\sum_{i\in [m]} \mu_{i,c} \quad \text{for all } c\in C'.\]
	\end{enumerate}
Note that $\mu_{i,c}$ depends as well on $\kappa$ in the last statement.
In \S\ref{sec:Realization_conditions} we \blue{mention strategies to check whether these realization conditions are satisfied,  and determine types of intermediates that satisfy them}.
		The proof of the next theorem is found in \cite{Feliu-Simplifying}. The first part is Theorem 5.1, and the second part is discussed in the text.
	
	\begin{thm}[\cite{Feliu-Simplifying}]\label{Theorem Multistationartity and Realization} 
		Let $\widetilde{\cN}$ be an extension of $\cN$ via the addition of intermediates $Y_1,\dots,Y_m$.
		\begin{enumerate}[(i)]
			\item 
			If the realization condition holds, then multistationarity of $\cN$ implies multistationarity of $\widetilde{\cN}$.
			\item Let $C\subseteq\cC$ be the set of inputs of $Y_1,\dots,Y_m$. If the generalized realization condition holds for $\widetilde{\cN}$, then $\widetilde{\cN}$ is multistationary if and only if the canonical extension $\widetilde{\cN}_C$ is multistationary.
		\end{enumerate}
	\end{thm}
	
	\begin{definition}\label{Definition Canonical Classes}
		Let $\cN$ be a reaction network and $C\subseteq\cC$. The \emph{canonical class} associated with $C$ is the set of all extensions of $\cN$ via the addition of intermediates with input set $C$  that satisfy the generalized realization condition.
	\end{definition}

	\begin{proposition}\label{Proposition Canonical Networks and Generalized Realization Condition}
		The generalized realization condition holds for canonical extensions. Therefore, a canonical class is not empty.
	\end{proposition}
	\begin{proof} Let $\widetilde{\cN}$ be the canonical extension of a network $\cN$  associated with  $C=\{c_1,\dots,c_m\}\subseteq\cC$. For every $i\in[m]$, we have $\mu_{i, c} =\tfrac{\k_{c_i\rightarrow Y_i}}{\k_{Y_i\rightarrow c_i}}$ if  $ c=c_i$ and zero otherwise.
		Since no two intermediates have a common input, the generalized realization condition holds if for every $k\in\mathbb{R}_{>0}^r$ and $\alpha \in\mathbb{R}_{>0}^m$, there exists $\kappa\in\mathbb{R}_{>0}^{\widetilde{r}}$ such that
\[
		k_{c\rightarrow c'} =\k_{c\rightarrow c'} \quad\textrm{for all }  c\rightarrow c'\in\cR \quad\textrm{and}\quad
		\alpha_i =\tfrac{\k_{c_i\rightarrow Y_i}}{\k_{Y_i\rightarrow c_i}}  \quad\textrm{for all }   i\in[m].
\]
		This condition clearly holds.
	\end{proof}
	
	Theorem \ref{Theorem Multistationartity and Realization} implies that multistationarity of an extended network $\widetilde{\cN}$ satisfying the generalized realization condition is equivalent to multistationarity of any network in the same canonical class of $\widetilde{\cN}$, in particular of the canonical extensions in the class.
	Canonical extensions have a simple structure and preserve some important properties of $\cN$ as we will see below. Hence they are chosen as representatives of the class. 
	
	\begin{example}\label{Example:Canonical_Representatives_MPAK}
\label{Example:Graph_MPAK}
		The following digraph defines  a reaction network   corresponding to the Mitogen-Activated Protein Kinase cascade \cite{MAPK-Multistationarity}:
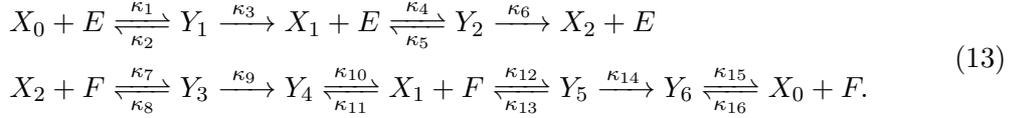
\begin{equation}\label{eq:MAPK}
\begin{aligned}
		X_0+E & \ce{<=>[\kappa_1][\kappa_2]} Y_1\ce{->[\kappa_3]} X_1+E\ce{<=>[\kappa_4][\kappa_5]} Y_2\ce{->[\kappa_6]} X_2+E\\
		X_2+F & \ce{<=>[\kappa_7][\kappa_8]} Y_3\ce{->[\kappa_9]} Y_4\ce{<=>[\kappa_{10}][\kappa_{11}]} X_1+F\ce{<=>[\kappa_{12}][\kappa_{13}]} Y_5\ce{->[\kappa_{14}]} Y_6\ce{<=>[\kappa_{15}][\kappa_{16}]} X_0+F.	
	\end{aligned}
\end{equation}
If we consider $Y_1,\dots,Y_6$ as intermediates, then the associated core network is
		\[X_0+E\ce{->[k_1]} X_1+E\ce{->[k_2]} X_2+E \qquad 
		X_2+F\ce{->[k_3]} X_1+F\ce{->[k_4]} X_0+F.\]
		The generalized realization condition holds  by Example~\ref{Example:Generalized_Realization_Condition_MAPK_Nsite} in \S\ref{sec:Realization_conditions}. 
		The canonical extension in the canonical class of  \eqref{eq:MAPK} is the following network:
		
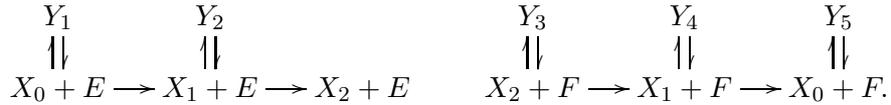
\[\xymatrix @C=1.25pc @R=1pc{ Y_1\ar@<+2pt>@{-^>}[d]\ar@<-2pt>@{_<-}[d] & Y_2\ar@<+2pt>@{-^>}[d]\ar@<-2pt>@{_<-}[d] & \\
			X_0+E\ar[r] & X_1+E\ar[r] & X_2+E}
		\qquad
		\xymatrix @C=1.25pc @R=1pc{
			Y_3\ar@<+2pt>@{-^>}[d]\ar@<-2pt>@{_<-}[d] & Y_4\ar@<+2pt>@{-^>}[d]\ar@<-2pt>@{_<-}[d] &
			Y_5\ar@<+2pt>@{-^>}[d]\ar@<-2pt>@{_<-}[d] \\
			X_2+F\ar[r] & X_1+F\ar[r] & X_0+F.}\]
	\end{example}
	
\subsection{Binomial networks and intermediates}
	In this subsection, building on results from \cite{Grobner}, we relate the condition (det) for the core and extended network. This leads to a determinant criterion for multistationarity of extended networks with a complete binomial core network.
	
	Let  $\widetilde{\cN}$ be an extension of a binomial reaction network $\cN$   via the addition of intermediates $Y_1,\dots,Y_m$. \blue{For a fix monomial order and an ordered binomial basis $B$ of the steady state ideal  $I\subseteq\R(k)[x]$ of $\cN$}, let $\Rem(f,B)$ denote the remainder of the division of a polynomial $f$ by $B$.
	Assume $\Phi(B)$ is defined and consider 
	\begin{align}
\widetilde{B} & =  \Phi(B)\cup\Big\{y_i-\sum_{c\in\cC}\mu_{i,c}x^c , i\in [m] \Big\},   \label{eq:Btilde} \\
 \widetilde{B}' & = \Phi(B)\cup\Big\{y_i-\Rem\Big(\sum_{c\in\cC}\mu_{i,c}x^c,\Phi(B)\Big), i\in [m]\Big\}.   \label{eq:Btildeprime}
\end{align}
If the set on the right-hand side of the union in either $\widetilde{B}$ or $ \widetilde{B}' $ consists of binomials, then they have the form	\begin{equation}\label{Equation:Binomial_Basis_for_extended_networks}
\Phi(B)\cup \big\{y_i-p_i(\k)x^{\alpha_i},\  i\in [m]\big\}
		\end{equation}
where $p_i(\k)\in \R(\k)$ and $\alpha_i$ is a vector of non-negative integers.
If all intermediates are 1-input,  then a binomial basis of $\widetilde{I}$ is
	\begin{equation}\label{Equation:Binomial_Basis_for_extended_networks_1_input}
	\widetilde{B}=\Phi(B)\cup\{y_i-\mu_{i,c}x^{c_i},\ i\in [m]\},
	\end{equation}
where $c_i$ is the only input of $Y_i$, and this basis is  admissible provided $B$ is an admissible binomial basis of $I$ and $\Phi(B)$ is defined. This applies in particular to canonical extensions.
	
	\begin{definition}
		Let $\cN$ be a binomial reaction network and $\widetilde{\cN}$ an extension of it via the addition of intermediates. $\widetilde{\cN}$ is a \emph{binomial extension} of $\cN$ if
		 there exists an admissible binomial basis $B$ of $\cN$ such that  $\Phi(B)$ is well defined, no coefficient of $B$ becomes zero under $\Phi$, and further, either $\widetilde{B}$  or $ \widetilde{B}' $ is an admissible binomial basis of $\widetilde{\cN}$. \blue{In this case we say that $B$ and $\widetilde{B}$ are compatible binomial bases.}
	\end{definition}

\begin{remark}\label{remark:algind}
If the functions $\phi_{c\rightarrow c'}$ are algebraically independent over $\R$, then $\Phi(B)$ is well defined and no coefficient of $B$ vanishes. 	By \cite[Lemma 3.3]{Grobner}, \eqref{eq:Btilde} and \eqref{eq:Btildeprime} are bases of the steady state ideal of the extended network $\widetilde{I}\subseteq\R(\k)[x,y]$, and if $B$ is admissible, then so is $\widetilde{B}$.  To decide whether $\widetilde{B}'$ is also admissible, it suffices to check that the representations of $y_i-\Rem\big(\sum_{c\in\cC}\mu_{i,c}x^c,\Phi(G)\big)$ in terms of $\widetilde{B}$ and that of $y_i-\sum_{c\in\cC}\mu_{i,c}x^c$ in terms of $\widetilde{B}'$ are well defined for all $\k$ (c.f.~Lemma \ref{Lemma:Bases_having_same_solution_sets}).

Further, by  \cite[Theorem 3.10]{Grobner}, the steady state ideal $\widetilde{I}$ of  $\widetilde{\cN}$ is binomial if and only if the steady state ideal $I$ of $\cN$ is binomial and for any reduced Gr\"obner basis $G$ of $I$, $\Rem\big(\sum_{c\in\cC}\mu_{i,c}x^c,\Phi(G)\big)$ has at most one term for all $i\in[m]$. 
\end{remark}

\begin{lem}\label{Lemma:Rank_for_extentions}
		Let $\widetilde{\cN}$ be a binomial extension of  a binomial network $\cN$. Condition (rank) holds for $\widetilde{\cN}$ if and only if it holds for $\cN$.
	\end{lem}
	\begin{proof}
		Let $B$ and $\overline{B}$ be admissible binomial bases of the steady state ideals of $\cN$ and $\widetilde{\cN}$ respectively, such that $\overline{B}$ is either $\widetilde{B}$ in \eqref{eq:Btilde} or $\widetilde{B}'$ in \eqref{eq:Btildeprime}. 
Then   $|\overline{B}|=|B|+m$. If $n-d$ is the rank of $\cN$, then by \eqref{Eq:Matrix_of_Conservation_Laws_for_extended-networks} and the text below it, 
the condition (rank) for $\widetilde{\cN}$ is $n+m-d=|\overline{B}| =|B|+m $, which is the rank condition for $\cN$, 
 $|B|=n-d$.
	\end{proof}
It follows from the lemma  above that  a binomial extension of a complete binomial network satisfying (surj) is a complete binomial network.  In general, for an arbitrary extended network $\widetilde{\cN}$, we cannot guarantee that (surj) holds provided it holds for $\cN$. However, it does for canonical extensions.

\begin{proposition}\label{Lemma:Surj_for_Canonical_networks}
Let $\cN$ be a binomial network. Any canonical extension $\widetilde{\cN}_C$ of $\cN$ is a binomial extension. 
Further, (surj) holds for $\cN$ with an admissible binomial basis $B$ if and only if (surj) holds for $\widetilde{\cN}_C$ with $\widetilde{B}$ as in \eqref{eq:Btilde}.
Therefore,  a canonical extension of a complete binomial network is also complete.
\end{proposition}
	\begin{proof}
By \cite[Corollary 4.6]{Grobner},  the functions $\phi_{c\rightarrow c'}$ are algebraically independent for canonical extensions.  
By \eqref{Equation:Binomial_Basis_for_extended_networks_1_input} and Remark~\ref{remark:algind},  $\widetilde{\cN}_C$  is a binomial extension.

For the second part of the proposition, write $C=\{c_1,\dots,c_m\}$. 
 Denote the reaction rate constants of $\cN$ by $k_1,k_2,\dots,k_r$ following the order of the reaction set. The network $\widetilde{\cN}_C$ has  $r+2m$ reactions. We denote the reaction rate constants of the reactions of $\widetilde{\cN}_C$ that are also in $\cN$ with $\k_1,\dots,\k_r$ and of the other reactions by
 $c_i\ce{<=>[\k_{r+2i-1}][\k_{r+2i}]}Y_i$ for $i\in[m].$
Then $\mu_{i,c_i}= \tfrac{\kappa_{r+2i-1}}{\kappa_{r+2i}}$.
Let $M$ be the matrix constructed  in \eqref{Equation Gamma and M} for the basis $B$. 
Now (surj) holds  for $\widetilde{\cN}_C$ if for every $(x,y)\in\R^{n+m}_{>0}$ there exists $(\k_1,\dots,\k_{r+2m})\in\R^{r+2m}_{>0}$ such that
$x^M=\gamma(\kappa)$ and $\tfrac{y_i}{x^{c_i}}=\tfrac{\kappa_{r+2i-1}}{\kappa_{r+2i}}$ for all $ i\in[m]$.
		Since the first part of the system does not depend on $\kappa_{r+1},\dots,\kappa_{r+2m}$, the second part is always satisfied independently from the first part by letting $\kappa_{r+2i}=x^{c_i}$ and $\kappa_{r+2i-1}=y_i$. The first part of the system is exactly the same as (surj) for $\cN$ after replacing $k_i$ by $\kappa_i$. Hence (surj) holds for $\widetilde{\cN}_C$ if and only if it holds for $\cN$.
The last statement follows from Lemma \ref{Lemma:Rank_for_extentions}.
	\end{proof}

 Recall that criterion (det)  can be used to determine multistationarity for complete \blue{binomial} networks, c.f.~Theorem \ref{Theorem:Multistationarity_and_Determinant_Condition}.
	Consider  a binomial extension $\widetilde{\cN}$ of a  complete binomial network $\cN$.  The steady state ideal of $\widetilde{\cN}$ has an admissible binomial basis $\widetilde{B}$ of the form \eqref{Equation:Binomial_Basis_for_extended_networks}, with $B$ an admissible binomial basis of the steady state ideal of $\cN$. Let $M\in\R^{n\times s}$ be the exponent matrix associated with the binomials in $B$, c.f.~\eqref{Equation Gamma and M}. Since no coefficient of $B$ becomes zero under $\phi$, the exponents of the monomials in $B$ and $\Phi(B)$ agree. Hence the exponent matrix $\widetilde{M}$ associated with the binomials in $\widetilde{B}$ has the following form:
\[
(\widetilde{M})^T=\left[\begin{array}{c;{2pt/2pt}c}
	M^T & 0\\
	\hdashline[2pt/2pt]
	\begin{array}{c}
	-\alpha_1\\
	\vdots\\
	-\alpha_m
	\end{array} & I_m
	\end{array}\right]\in\R^{(s+m)\times(n+m)}.\]
	Using  \eqref{Eq:Matrix_of_Conservation_Laws_for_extended-networks}, the corresponding matrices $\Gamma$ and $\widetilde{\Gamma}$ in \eqref{eq:Gamma}, for $\cN$ and  $\widetilde{\cN}$ respectively, are:
\begin{equation}\label{eq:Gamma_tilde}
\Gamma=\begin{bmatrix}
	\blue{(M^T)_\lambda}\, \\ Z
	\end{bmatrix}\in\R^{n\times n}, \qquad \widetilde{\Gamma}=\left[\begin{array}{c;{2pt/2pt}c}
	\blue{(M^T)_\lambda} & 0\\
	\hdashline[2pt/2pt]
	\begin{array}{c}
	-\alpha_1\\
	\vdots\\
	-\alpha_m
	\end{array} & \begin{array}{ccc}
	\lambda_{n+1} & & 0\\
	& \ddots & \\
	0 & & \lambda_{n+m}
	\end{array}\\ 
	\hdashline[2pt/2pt]
	Z & \begin{array}{ccc} 
	Zc_1^T & \dots & Zc_m^T
	\end{array}
	\end{array}\right]\in\R^{(n+m)\times(n+m)},
\end{equation}
 where $Z$ is a matrix of conservation laws for $\cN$, $c_1,\dots,c_m$ are chosen inputs for $Y_1,\dots,Y_m$  and $\lambda=(\lambda_1,\dots,\lambda_n)$. 
	If $\widetilde{\cN}_C$ is a canonical extension, then
$\alpha_i=c_i$ in $\widetilde{\Gamma}$,  c.f.~\eqref{Equation:Binomial_Basis_for_extended_networks_1_input}, in which case we denote the matrix by $\widetilde{\Gamma}_C$.

The results in this subsection combined with Theorem \ref{Theorem Multistationartity and Realization}(ii) imply that we can use  a determinant condition to detect multistationarity of networks that are not necessarily binomial, but such that the core network is binomial and complete. 
	
	\begin{thm}[Determinant criterion for extensions of complete binomial networks]\label{Theorem:Usage_of_Canonical_networks_for_multistationarity}
		Let $\cN$ be a complete binomial network and $\widetilde{\cN}$ an extended network in the canonical class associated with $C\subseteq \cC$. Then $\widetilde{\cN}$ is multistationary if and only if \blue{(det) holds for the canonical extension $\widetilde{\cN}_C$.}
	\end{thm}
\begin{proof}
By Theorem \ref{Theorem Multistationartity and Realization}(ii), $\widetilde{\cN}$ is multistationary if and only if $\widetilde{\cN}_C$ is. \blue{Since $\widetilde{\cN}_C$ is a complete binomial network, it is multistationary if and only if (det) holds by Theorem \ref{Theorem:Multistationarity_and_Determinant_Condition}.}
\end{proof}
	
	\begin{example}\label{Example MAPK Multistationarity Structure}(Continued from Example \ref{Example:Canonical_Representatives_MPAK})
The network in Example \ref{Example:Canonical_Representatives_MPAK} is not binomial by \cite[Example 3.13]{Grobner},
but the core network is a complete binomial network. The extended network belongs to the canonical class associated with  $C=\{X_0+E,X_1+E,X_2+F,X_1+F,X_0+F\}$. With a suitable choice of basis $B$, the matrix $\widetilde{\Gamma}_C$ is as follows:
		{\small \[\left[\begin{array}{ccccc;{2pt/2pt}ccccc}
		-\lambda_1 & \lambda_2 & 0 & -\lambda_4 & \lambda_5 &  &  & 0 &  & \\
		0 & -\lambda_2 & \lambda_3 & -\lambda_4 & \lambda_5 &  &  &  &  & \\
		\hdashline[2pt/2pt]
		-\lambda_1 & 0 & 0 & -\lambda_4 & 0 & \lambda_6 &  &  &  & 0\\
		0 & -\lambda_2 & 0 & -\lambda_4 & 0 &  & \lambda_7 &  &  & \\
		0 & 0 & -\lambda_3 & 0 & -\lambda_5 &  &  & \lambda_8 &  & \\
		0 & -\lambda_2 & 0 & 0 & -\lambda_5 &  &  &  & \lambda_9 & \\
		-\lambda_1 & 0 & 0 & 0 & -\lambda_5 & 0 &  &  &  & \lambda_{10}\\
		\hdashline[2pt/2pt]
		1 & 1 & 1 & 0 & 0 & 1 & 1 & 1 & 1 & 1\\
		0 & 0 & 0 & 1 & 0 & 1 & 1 & 0 & 0 & 0\\
		0 & 0 & 0 & 0 & 1 & 0 & 0 & 1 & 1 & 1
		\end{array}\right].\]}The polynomial $\det(\widetilde{\Gamma}_C)$ has terms with different signs. Therefore by Theorem  \ref{Theorem:Usage_of_Canonical_networks_for_multistationarity}, the network in \eqref{eq:MAPK} is multistationary.
	\end{example}

	\section{Lifting multistationarity and the multistationarity structure}\label{sec:Lifting_multistationarity_and_the_multistationarity_structure}
	
In this section we use the following notation: for $J\subseteq [n]$ and $\lambda$  an  $n$-tuple of indeterminates/numbers, we define $\lambda_J=\prod_{i\in J}\lambda_i$.

	\subsection{Lifting multistationarity}
Theorem  \ref{Theorem Multistationartity and Realization} tells us that multistationarity of the core network implies multistationarity of the extended network if the realization condition is satisfied. In this scenario we informally say that multistationarity is lifted. 
In this subsection we show that multistationarity is   lifted for binomial extensions of complete binomial networks, even if the realization condition is not satisfied \blue{or we might not be able to verify that it holds}.
Before that, we start with a lemma on the structure of $\widetilde{\Gamma}$. 

	\begin{lem}\label{Lemma:Determinant_of_Gamma_Tilde}
		Let  $\widetilde{\cN}$ be a binomial extension  of a complete binomial network $\cN$ via the addition of  intermediates $Y_1,\dots,Y_m$,  \blue{let $B$ and $\widetilde{B}$ be compatible binomial bases that satisfy (rank) and let $\Gamma,\widetilde{\Gamma}$ be derived as in \eqref{eq:Gamma_tilde} for this choice of binomial bases.} Then
		\[\det(\widetilde{\Gamma})=\lambda_{[n+1,n+m]}\det(\Gamma)+p'(\lambda),\]
		where $p'(\lambda)$ is a polynomial in $\lambda$ such that none of its terms is divisible by $\lambda_{[n+1,n+m]}$.
	\end{lem}
	\begin{proof}
		Let $s=\rank(N)$ and $\widetilde{\Gamma}_{[s+1,s+m],[n+1,n+m]}$ be the submatrix of $\widetilde{\Gamma}$ obtained by removing the rows with index in $[s+1,s+m]$ and the columns with index in $[n+1,n+m]$. By the generalized Laplacian expansion of $\det(\widetilde{\Gamma})$ along rows $s+1,\dots,s+m$ we have that
		\[\begin{array}{ll}
		\det(\widetilde{\Gamma}) & =\lambda_{[n+1,n+m]}\det\big(\widetilde{\Gamma}_{[s+1,s+m],[n+1,n+m]}\big)+p'(\lambda)  =\lambda_{[n+1,n+m]}\det(\Gamma)+p'(\lambda),
		\end{array}\]
		where $p'(\lambda)$ is a polynomial in $\lambda$. By construction, $p'(\lambda)$ does not have any monomial multiple of $\lambda_{[n+1,n+m]}$. 
	\end{proof}

\begin{thm}[Lifting multistationarity]\label{Proposition:Multistationarity_and_binomial_extensions}
Let  $\widetilde{\cN}$ be a binomial extension  of a complete binomial network $\cN$ via the addition of $m$ intermediates $Y_1,\dots,Y_m$,   \blue{let $B$ and $\widetilde{B}$ be compatible  binomial bases such that (rank) and (surj) hold, and let $\Gamma$ be as in \eqref{eq:Gamma} for $B$. } If $\cN$ is multistationary and $\det(\Gamma)\neq 0$, then $\widetilde{\cN}$ is multistationary.
	\end{thm}
	\begin{proof} Let  $\lambda=(\lambda_1,\dots,\lambda_{n+m})$ and $\bar{\lambda}=(\lambda_1,\dots,\lambda_{n})$. 
		Since $\cN$ is multistationary   and by hypothesis $\det(\Gamma)\neq 0$, by Theorem \ref{Theorem:Multistationarity_and_Determinant_Condition}  we have that $\det(\Gamma)$ is a polynomial in $\lambda_1,\dots,\lambda_n$ with two terms of different non-zero sign, namely $\alpha\bar{\lambda}^u$ and $-\beta\bar{\lambda}^{v}$ with $\alpha,\beta>0$ and $u,v\in \Z_{\geq 0}^n$. 
\blue{Now consider the matrix $\widetilde{\Gamma}$ from \eqref{eq:Gamma_tilde} using the basis $\widetilde{B}$.}
By Lemma \ref{Lemma:Determinant_of_Gamma_Tilde}, $\det(\widetilde{\Gamma})$ has two terms with different non-zero sign $\alpha\bar{\lambda}^u\lambda_{[n+1,n+m]}$ and $-\beta\bar{\lambda}^v\lambda_{[n+1,n+m]}$. Since $\widetilde{\cN}$ is a complete binomial network,  $\widetilde{\cN}$ is multistationary by Theorem \ref{Theorem:Multistationarity_and_Determinant_Condition}.
	\end{proof}
	
We finish this subsection with \blue{a couple of examples}
 where Theorem \ref{Proposition:Multistationarity_and_binomial_extensions} allows us to conclude that an extended network is multistationary, while the realization condition is \blue{either} not satisfied \blue{or not easy to verify} (and hence Theorem \ref{Theorem Multistationartity and Realization} cannot be applied).
	
	\begin{example}\label{Example:Surjectivity_Yes_Realisation_No} 
		Consider the following network $\cN$:
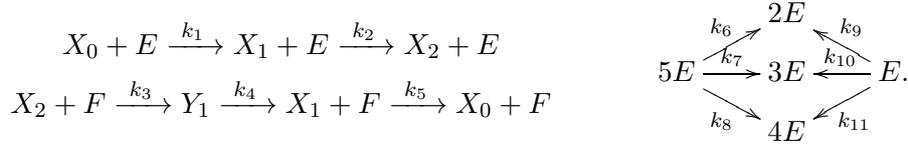
		\[
\begin{array}{cc}
		\adjustbox{valign=c}{$\begin{array}{c}
			X_0+E\ce{->[k_1]}X_1+E\ce{->[k_2]}X_2+E\\[5pt]
			X_2+F\ce{->[k_3]}Y_1\ce{->[k_4]}X_1+F\ce{->[k_5]}X_0+F
			\end{array}$} & \qquad
		\adjustbox{valign=c}{\xymatrix @C=1.75pc @R=0.75pc { & 2E & \\
				5E\ar[ur]^{k_6}\ar[r]^{k_7}\ar[dr]_{k_8} & 3E & E.\ar[ul]_{k_9}\ar[l]_{k_{10}}\ar[dl]^{k_{11}}\\
				& 4E & }}
		\end{array}
		\]
$\cN$ has rank $4$ and the steady state ideal has the following admissible binomial basis with $4$ elements: $B=\{  -k_1x_0e+k_5x_1f, -k_2x_1e+k_3k_4x_2f, k_4y_1-k_3x_2f, (k_{9}+2k_{10}+3k_{11})e-(3k_6+2k_7+k_8)e^5\}$.  
		Since $(1,1,1,1,1,4,5,6,6,5,4)\in\ker(N)\cap\R_{>0}^{11}$, (surj) holds. Therefore $\cN$  is a complete binomial network. Using (det), we see that $\cN$ is multistationary and $\det(\Gamma)\neq 0$.
		Now consider the following extension $\widetilde{\cN}$ of $\cN$ via the addition of one intermediate $Y_2$:
		\[
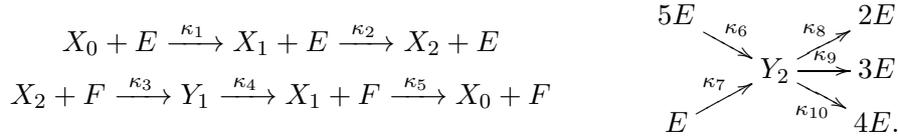
\begin{array}{cc}
		\adjustbox{valign=c}{$\begin{array}{c}
			X_0+E\ce{->[\k_1]}X_1+E\ce{->[\k_2]}X_2+E\\[5pt]
			X_2+F\ce{->[\k_3]}Y_1\ce{->[\k_4]}X_1+F\ce{->[\k_5]}X_0+F
			\end{array}$} & \qquad
		\adjustbox{valign=c}{\xymatrix @C=1.5pc @R=0.5pc { 5E\ar[dr]^{\k_6} & & 2E\\
				& Y_2\ar[ur]^{\k_8}\ar[r]^{\k_9}\ar[dr]_{\k_{10}} & 3E \\
				E\ar[ur]^{\k_7} & & 4E.}}
		\end{array}
		\]
This network does not satisfy the realization condition (see \S 2.2 in the electronic supplementary material of \cite{Feliu-Simplifying}).
The  well-defined set $\widetilde{B}'=\Phi(B) \cup \{y_2-\tfrac{4\k_7e}{3\k_8+2\k_9+\k_{10}}\}$ in \eqref{eq:Btildeprime} makes $\widetilde{\cN}$  a binomial extension.
Further $(1,1,1,1,1,5,5,4,2,4)\in\ker(\widetilde{N})\cap\R_{>0}^{10}$, and hence (surj) holds for $\widetilde{N}$.
By  Theorem \ref{Proposition:Multistationarity_and_binomial_extensions}, we conclude that $\widetilde{\cN}$ is also multistationary. 
	\end{example}

\color{black}
\begin{example}\label{Example:XX}
In this example we illustrate that for some examples, the determinant of $\Gamma$ can be much smaller than the determinant of $\widetilde{\Gamma}$. 
Consider the following reaction network $\widetilde{\cN}$:
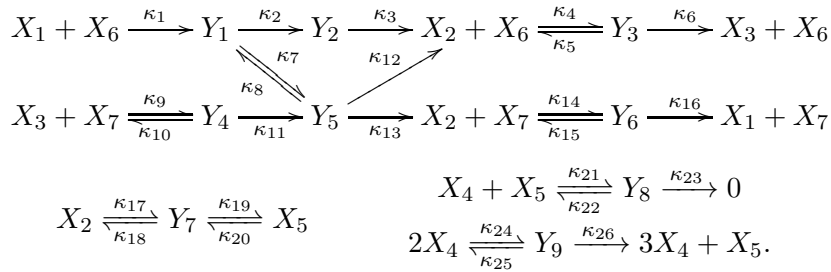
\[\adjustbox{valign=c}{\xymatrix @C=2pc @R=1.5pc { X_1+X_6\ar[r]^{\quad \kappa_1} & Y_1\ar[r]^{\k_2}\ar@<+1pt>@{-^>}[dr]^{\kappa_7}\ar@<-1pt>@{_<-}[dr]_{\kappa_8} & Y_2\ar[r]^{\k_3\quad} & X_2+X_6\ar@<+1pt>@{-^>}[r]^{\quad \kappa_4}\ar@<-1pt>@{_<-}[r]_{\quad \kappa_5} & Y_3\ar[r]^{\k_6\quad} & X_3+X_6 \\
		X_3+X_7\ar@<+1pt>@{-^>}[r]^{\quad \kappa_9}\ar@<-1pt>@{_<-}[r]_{\quad \kappa_{10}} & Y_4\ar[r]_{\k_{11}} & Y_5\ar[ur]^{\k_{12}}\ar[r]_{\k_{13}\quad} & X_2+X_7\ar@<+1pt>@{-^>}[r]^{\quad \kappa_{14}}\ar@<-1pt>@{_<-}[r]_{\quad \kappa_{15}} & Y_6\ar[r]^{\k_{16}\quad} & X_1+X_7 }
                    }\]

\[        \begin{array}{cc}
        X_2\ce{<=>[\k_{17}][\k_{18}]} Y_7\ce{<=>[\k_{19}][\k_{20}]} X_5 & \qquad 
\adjustbox{valign=c}{$\begin{array}{c}
                            			X_4+X_5\ce{<=>[\k_{21}][\k_{22}]} Y_8\ce{->[\k_{23}]}0\\[5pt]
                            			2X_4\ce{<=>[\k_{24}][\k_{25}]} Y_9\ce{->[\k_{26}]}3X_4+X_5.
                            			\end{array}$}
        \end{array}
			\] 
The species $Y_1,\dots,Y_9$ are intermediates, which, upon removal, lead to the following core network $\cN$:
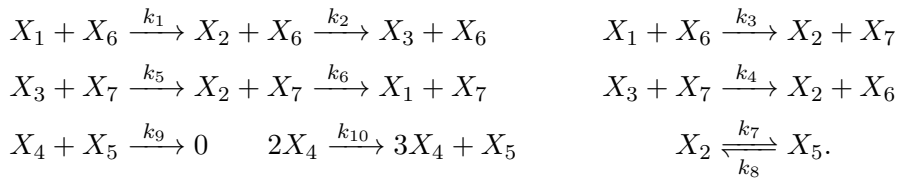
\begin{align*}
X_1+X_6  & \ce{->[k_1]}  X_2+X_6 \ce{->[k_2]}  X_3+X_6  & X_1+X_6  & \ce{->[k_3]} X_2+X_7 &     \\ 
X_3+X_7 & \ce{->[k_5]}  X_2+X_7 \ce{->[k_6]}  X_1+X_7  & X_3+X_7 & \ce{->[k_4]}X_2+X_6   \\
  X_4+X_5 & \ce{->[k_9]}0  \qquad  2X_4  \ce{->[k_{10}]}3X_4+X_5  & X_2 & \ce{<=>[k_7][k_8]} X_5.
           \end{align*}
The realization condition for $\widetilde{\cN}$ requires, for any given $k\in \R_{>0}^4$, the existence of $\k\in \R_{>0}^{10}$ that satisfies the following identities:
\begin{align*}
k_1 & =\tfrac{\k_2(\k_8+\k_{12}+\k_{13})+\k_1\k_7}{(\k_2+\k_7)(\k_8+\k_{12}+\k_{13})-\k_7\k_8} & 
k_3& =\tfrac{\k_1\k_7\k_{13}}{(\k_2+\k_7)(\k_8+\k_{12}+\k_{13})-\k_7\k_8}\\[6pt]
k_4& =\tfrac{\k_2\k_8\k_9\k_{11}+(\k_2+\k_7)\k_9\k_{11}\k_{12}}{(\k_{10}+k_{11})\big((\k_2+\k_7)(\k_8+\k_{12}+\k_{13})-\k_7\k_8\big)} & 
k_5 & =\tfrac{(\k_2+\k_7)\k_9\k_{11}\k_{13}}{(\k_{10}+k_{11})\big((\k_2+\k_7)(\k_8+\k_{12}+\k_{13})-\k_7\k_8\big)}.
\end{align*}
This condition cannot be verified by means of  Proposition \ref{Lemma:Realization_more_classes}. 
We turn instead to Theorem~\ref{Proposition:Multistationarity_and_binomial_extensions}, and to this end, we argue that $\cN$ is a complete binomial network, $\widetilde{\cN}$ a binomial extension, and that there exist compatible bases $B$ and $\widetilde{B}$ that satisfy (rank) and (surj). 
The rank of $\cN$ is $5$. An admissible binomial basis $B$ for $\cN$ that satisfies (rank) is found by performing linear combinations of the steady state polynomials:
\begin{multline*}
\big\{-k_3x_1x_6+k_4x_3x_7,\,-(k_1+k_3)x_1x_6+k_6x_2x_7,\,k_2x_2x_6-(k_4+k_5)x_3x_7,\\ \,k_7x_2-k_8x_5,-k_9x_4x_5+k_{10}x_4^2\big\}.
\end{multline*}
Since $(1,2,1,1,1,2,1,1,1,1)\in\ker(N)\cap\R_{>0}^{10}$,   (surj) holds. Therefore, $\cN$ is a complete binomial network.
The matrix $\Gamma$ constructed from $B$ is 
\[\Gamma= \begin{bmatrix}
\lambda_1 & 0 & -\lambda_3 & 0 & 0 & \lambda_6 & -\lambda_7\\
\lambda_1 & -\lambda_2 & 0 & 0 & 0 & \lambda_6 & -\lambda_7\\
0 & \lambda_2 & -\lambda_3 & 0 & 0 & \lambda_6 & -\lambda_7\\
0 & \lambda_2 & 0 & 0 & -\lambda_5 & 0 & 0\\
0 & 0 & 0 & -\lambda_4 & \lambda_5 & 0 & 0\\
0 & 0 & 0 & 0 & 0 & 1 & 1\\
1 & 1 & 1 & -1 & 1 & 0 & 0
\end{bmatrix}, \]
and its determinant is $ (\lambda_1\lambda_2\lambda_3\lambda_4-\lambda_1\lambda_2\lambda_3\lambda_5+\lambda_1\lambda_2\lambda_4\lambda_5+\lambda_1\lambda_3\lambda_4\lambda_5+\lambda_2\lambda_3\lambda_4\lambda_5)(\lambda_6+\lambda_7).$
Since it has terms of both signs, $\cN$ is multistationary by Theorem~\ref{Theorem:Multistationarity_and_Determinant_Condition}.

We check the conditions on $\widetilde{\cN}$ now. It is straightforward to see that $\Phi(B)$ is well defined. 
Let $g=\k_2\k_8+\k_2\k_{12}+\k_2\k_{13}+\k_7\k_{12}+\k_7\k_{13}$. Then, an admissible basis $\widetilde{B}$ for the steady state ideal of the extended network, found as in \eqref{eq:Btilde}, consists of the following polynomials
\begin{align*}
& \quad -\phi(k_3)x_1x_6+\phi(k_4)x_3x_7,\quad -\big(\phi(k_1)+\phi(k_3)\big)x_1x_6+\phi(k_6)x_2x_7,\\
& \phi(k_2)x_2x_6-\big(\phi(k_4)+\phi(k_5)\big)x_3x_7,\quad \phi(k_7)x_2-\phi(k_8)x_5, \quad -\phi(k_9)x_4x_5+\phi(k_{10})x_4^2,\\
& y_1-\tfrac{\k_8+\k_{12}+\k_{13}}{g}x_1x_6-\tfrac{\k_8\k_9\k_{11}}{(k_{10}+\k_{11})g}x_3x_7, \quad y_2-\tfrac{\k_2(\k_8+\k_{12}+\k_{13})}{\k_3g}x_1x_6-\tfrac{\k_2\k_8\k_9\k_{11}}{\k_3(k_{10}+\k_{11})g}x_3x_7,\\
& y_3-\tfrac{\k_4}{\k_5+\k_6}x_2x_6, \quad y_4-\tfrac{\k_9}{\k_{10}+\k_{11}}x_3x_7,\quad y_5-\tfrac{\k_1\k_7}{g}x_1x_6-\tfrac{\k_9\k_{11}(\k_2+\k_7)}{(\k_{10}+\k_{11})g}x_3x_7,\\
& y_6-\tfrac{\k_{14}}{\k_{15}+\k_{16}}x_2x_7,\quad y_7-\tfrac{\k_{17}}{\k_{18}+\k_{19}}x_2-\tfrac{\k_{20}}{\k_{18}+\k_{19}}x_5, \quad y_8-\tfrac{\k_{21}}{\k_{22}+\k_{23}}x_4x_5,\quad y_9-\tfrac{\k_{24}}{\k_{25}+\k_{26}}x_4^2.
\end{align*}
Only four of the polynomials are not binomial. We construct now the new admissible basis $\widetilde{B}'$ from \eqref{eq:Btildeprime}, and then the polynomials of the form $y_i-\Rem\Big(\sum_{c\in\cC}\mu_{i,c}x^c,\Phi(B)\Big)$ are all  binomial. 
We conclude that $B$ and $\widetilde{B}'$ are compatible binomial bases. The condition (surj) holds as well for $\widetilde{\cN}$  since 
\[(2,1,1,2,1,1,2,1,2,1,1,1,1,3,1,2,1,1,1,1,2,1,1,2,1,1)\in\ker(\widetilde{N})\cap\R_{>0}^{26}.\]
Therefore, we   apply Theorem \ref{Theorem:Multistationarity_and_Determinant_Condition} and conclude that $\widetilde{\cN}$ is multistationary. If instead, we compute $\det(\widetilde{\Gamma})$, we obtain a polynomial with $68$ terms, which has terms of both signs.
\end{example}
\color{black}
	
	\subsection{Multistationarity structure}
In this subsection we introduce the \emph{multistationarity structure} of a core network, consisting of the  subsets of complexes that  give rise to multistationarity when being the input of some intermediate. 
Note that   if $C_1\subseteq C_2\subseteq\cC$, then $\widetilde{\cN}_{C_2}$ is a binomial extension of $\widetilde{\cN}_{C_1}$. Using this and Lemma \ref{Lemma:Determinant_of_Gamma_Tilde}, we devise a strategy to determine the multistationarity structure of complete binomial networks by computing the determinant of $\widetilde{\Gamma}_{\cC}$ corresponding to the largest canonical extension. 	
We start with an example that illustrates the approach. 

	\begin{example}\label{Example:Explaining_how_to_use_det_to_study_canonical_models}
	 Consider the network in Example \ref{Example:A_Biological_Circuit_Determinant}, where $Y_1$ and $Y_2$ are intermediates. The associated core network is
    \begin{equation}\label{Equation:Core_network_of_simple_circuit}
    X_1\ce{<-[k_1]}0\ce{<-[k_3]}X_2\qquad X_1+E\ce{->[k_2]}X_2+E\qquad 2X_1+E,
    \end{equation}
	which gives $\cC=\{ 0,X_1,X_1+E,X_2+E,2X_1+E,X_2\}$. 
	An admissible binomial basis of the steady state ideal of \eqref{Equation:Core_network_of_simple_circuit} is $\{k_1-k_2x_1e,\,k_2x_1e-k_3x_2\}$, \blue{which is easily found by performing linear combinations of the steady state polynomials}. Since the rank of \eqref{Equation:Core_network_of_simple_circuit} is two and $(1,1,1)\in\ker(N)\cap\R_{>0}^3$, \eqref{Equation:Core_network_of_simple_circuit} is a complete binomial network. 
	The polynomial $\det(\widetilde{\Gamma}_\cC)$ for the largest canonical extension $\widetilde{\cN}_{\cC}$ is as follows:
{\small 	\begin{align*}
	\det(\widetilde{\Gamma}_\cC) &=\left|\begin{array}{ccc;{2pt/2pt}cccccc}
	-\lambda_1 & 0 & -\lambda_3 & & & & 0 & & \\
	\lambda_1 & -\lambda_2 & \lambda_3 & & & & & & \\
	\hdashline[2pt/2pt]
	0 & 0 & 0 & \lambda_4 & & & & & 0 \\
	-\lambda_1 & 0 & 0 & & \lambda_5 & & & & \\
	-\lambda_1 & 0 & -\lambda_3 & & & \lambda_6 & & & \\
	0 & -\lambda_2 & -\lambda_3 & & & & \lambda_7 & & \\
	-2\lambda_1 & 0 & -\lambda_3 & & & & & \lambda_8 & \\
	0 & -\lambda_2 & 0 & 0 & & & & & \lambda_9\\
	\hdashline[2pt/2pt]
	0 & 0 & 1 & 0 & 0 & 1 & 1 & 1 & 0
	\end{array}\right| = \begin{array}[t]{l} -\lambda_1\lambda_2\lambda_3\lambda_4\lambda_5\lambda_6\lambda_7\lambda_9 \\ \quad +\lambda_1\lambda_2\lambda_3\lambda_4\lambda_5\lambda_6\lambda_8\lambda_9 \\ \quad +\lambda_1\lambda_2\lambda_4\lambda_5\lambda_6\lambda_7\lambda_8\lambda_9.
\end{array}
	\end{align*}}%
Since \blue{$\widetilde{\cN}_{\cC}$} is a binomial extension of all canonical extensions, and all canonical extensions are complete, we can use Lemma \ref{Lemma:Determinant_of_Gamma_Tilde}  to find $\det(\widetilde{\Gamma}_C)$ for all $C\red{\subseteq}\cC$. For example, the $3$rd and $5$th complexes form the set $C=\{X_1+E,2X_1+E\}$ and $\det(\widetilde{\Gamma}_C)$ is the coefficient of $\lambda_4\lambda_5\lambda_7\lambda_9$ in $\det(\widetilde{\Gamma}_\cC)$: 
	\[\begin{array}{l}
	\det\hspace{-0.06cm}\big(\widetilde{\Gamma}_{\{X_1+E,2X_1+E\}}\big)=-\lambda_1\lambda_2\lambda_3\lambda_6+\lambda_1\lambda_2\lambda_6\lambda_8.
	\end{array}\]
For any subset $J\subseteq \{4,\dots,9\}$, the coefficient of $\lambda_J$ in $\det(\widetilde{\Gamma}_\cC)$ is $\det(\widetilde{\Gamma}_C)$ for the set of complexes $c_i$ such that $i+3\notin J$.

By Theorem \ref{Theorem:Usage_of_Canonical_networks_for_multistationarity}, all networks in the canonical class associated with $\{X_1+E,2X_1+E\}$ are multistationary.
In particular,  the network in Example \ref{Example:A_Biological_Circuit_Determinant} belongs to this canonical class (and hence is multistationary), since it satisfies the generalized realization condition:
	\[\begin{array}{l}
	\text{for every }(k,r)\in\R_{>0}^{3+2}\text{ there exists }\k\in\R_{>0}^7\text{ such that }\\[5pt] \qquad\qquad 
	k_1=\k_1,\quad k_2=\tfrac{\k_2\k_4}{\k_3+\k_4},\quad k_3=\k_7,\quad r_1=\tfrac{\k_2}{\k_3+\k_4},\quad r_2=\tfrac{\k_5}{\k_6}.
	\end{array}\]
	\end{example}
	
	Motivated by this example, we proceed as follows. 
\blue{Let $ \mathcal{P}(\cC)$ denote the power set of $\cC$, that is, the set of all subsets of $\cC$.}	

	\begin{definition}\label{Definition Multistationary structure}
		Let $\cN$ be a reaction network, with set of complexes $\cC$. Let \blue{$\Mult\subseteq \mathcal{P}(\cC)$} be the set of all subsets of complexes $C\subseteq\cC$ for which the canonical extension $\widetilde{\cN}_C$ of $\cN$ associated with $C$ is multistationary. 
Denote by $\Circuits$ the set of minimal elements of $\Mult$ with respect to inclusion.
The set $\Mult$ is called the \emph{multistationarity structure} of $\cN$ and the elements of $\Circuits$ are called the \emph{circuits of multistationarity} of $\cN$.
	\end{definition}
		
By  Theorem \ref{Theorem Multistationartity and Realization}(i), the sets $\Mult$ and $\Circuits$ associated with a complete binomial network determine each other.
\blue{We will use the following notation throughout the rest of this section. Given a complete binomial network $\cN$, we choose an admissible basis $B$ that satisfies (rank) and (surj). For the largest canonical extension $\cN_\cC$ of $\cN$, we consider the admissible basis $\widetilde{B}$ given in \eqref{Equation:Binomial_Basis_for_extended_networks_1_input}, which satisfies (rank) and (surj) and is compatible with $B$. Then, we let 
\begin{equation}\label{eq:DN}
D_\cN=\det(\widetilde{\Gamma}_{\cC}),
\end{equation} 
with $\widetilde{\Gamma}$ derived from this data as in \eqref{eq:Gamma_tilde}.
}
We assume that  the set of complexes is ordered $\cC=\{c_1,\dots,c_m\}$.
	
	\begin{lem}\label{Lemma:Mult_from_D}
Assume $\cN$ is a complete binomial network \blue{and $D_\cN$ obtained as in \eqref{eq:DN}}. Then 
		$C\in\Mult$ if and only if the coefficient of $\prod_{c_i\in \blue{\cC\setminus C}}\lambda_{n+i}$ in $D_\cN$ is zero or has two terms  with different non-zero sign.
	\end{lem}
	\begin{proof}
	We assume without loss of generality  that $C=\{c_1,\dots,c_t\}$ and $\blue{\cC\setminus C}=\{c_{t+1},\dots,c_m\}$. 
By Lemma \ref{Lemma:Determinant_of_Gamma_Tilde} applied to the complete binomial network $\widetilde{\cN}_C$ and the binomial extension $\widetilde{\cN}_\cC$, we have
\[D_\cN=\det(\widetilde{\Gamma}_C)\lambda_{[n+t+1,n+m]}+p'(\lambda), \]
  where $\lambda_{[n+t+1,n+m]}$ does not divide any term of $p'$. Hence  $\det(\widetilde{\Gamma}_C)$ is the coefficient of $\lambda_{[n+t+1,n+m]}$ in $D_\cN$. The statement now follows by Theorem \ref{Theorem:Multistationarity_and_Determinant_Condition}.
	\end{proof}

\blue{	
If $\cN$ is not multistationary, then all coefficients of $\det(\Gamma)$ have the same sign, which agrees with the sign of the coefficients of  $\lambda_{[n+1,n+t]}$ in $D_\cN$. This yields the following lemma.
}

	\begin{lem}\label{Lemma:Mult_from_D_fixed_sign}
		Let $\cN$ be a complete binomial network that is not multistationary \blue{and $D_\cN$ obtained as in \eqref{eq:DN}}. Then $C\in\Mult$ if and only if a term of $D_\cN$  is a multiple of $\prod_{c_i\in\blue{\cC\setminus C}}\lambda_{n+i}$ with sign different than the sign of multiples of $\lambda_{[n+1,n+t]}$ in $D_\cN$.
	\end{lem}
	\begin{proof}
Assume $C=\{c_1,\dots,c_t\}$ and $\blue{\cC\setminus C}=\{c_{t+1},\dots,c_m\}$. We consider $\cN$, $\widetilde{\cN}_C$ and $\widetilde{\cN}_{\cC}$, and note that each network is a binomial extension of the previous. We apply  Lemma \ref{Lemma:Determinant_of_Gamma_Tilde} twice and obtain
		\begin{align*}
		D_\cN =& \det(\widetilde{\Gamma}_C)\lambda_{[n+t+1,n+m]}+p'= \big(\det(\Gamma)\lambda_{[n+1,n+t]}+p''\big)\lambda_{[n+t+1,n+m]}+p'.
		\end{align*}
		Since $\cN$ is not multistationary, $\det(\Gamma)$  is non-zero  and all terms have the same sign by Theorem \ref{Theorem:Multistationarity_and_Determinant_Condition}. Hence the coefficient of $\lambda_{[n+1,n+m]}$ is a polynomial where all coefficients have the same sign, $\tau$.
		Now by Lemma \ref{Lemma:Mult_from_D}, $C\in\Mult$ if and only if  the coefficient of $\lambda_{[n+t+1,n+m]}$ has  terms with different sign. Therefore $C\in\Mult$ if and only if there is a term in $p''$ with sign $-\tau$. This proves the lemma.
	\end{proof}
	
We are now ready to  introduce an algorithm to determine the multistationarity structure of a complete binomial network. If the  network is multistationary  (\blue{decidable by computing $\det(\Gamma)$}),  then all canonical extensions are multistationary, \blue{thus $\Mult=\mathcal{P}(\cC)$,} and hence $\Circuits=\{\emptyset\}$.  \blue{If the largest canonical extension is not multistationary, then $\Mult=\Circuits=\emptyset$.}
	
	\begin{algorithm}[\blue{Multistationarity structure for  complete binomial networks}]\label{Algorith:Circuits_Main}
		\hfill
		\begin{itemize}
		\setlength\itemsep{0em}
			\item [] \emph{Input:} A complete binomial network $\cN$.
			\item [] \emph{Output:} $\Circuits$ for $\cN$.
			\item [] \emph{Procedure:}
			\begin{itemize}
             \item       \blue{      Compute $\det(\Gamma)$. }
				\item        \blue{If $\det(\Gamma)$ is either zero or has coefficients with different sign, then return  $\Circuits=\{\emptyset\}$.}
                         \item \blue{ Otherwise}
					\begin{itemize}
				\item [] \emph{Initialize:} $\Circuits=\emptyset$.
				\item [] \emph{1.} Compute $D_\cN$ and let $\tau$ be the sign of any term divisible by $\lambda_{[n+1,n+m]}$.
				\item [] \emph{2.} For every term $T=\alpha \lambda^u$ of $D_\cN$ with sign $-\tau$:
				\begin{itemize}
						 \item [] \emph{2a.} Define $C=\{c_i\mid u_{n+i}=0,\ i=1,\dots,m\}$. 
						 \item [] \emph{2b.} Add $C$ to $\Circuits$ if no subset of $C$ is already in $\Circuits$,  and subsequently remove the supersets of $C$ from $\Circuits$, if any.
				\end{itemize}
			\end{itemize}
	\end{itemize}
		\end{itemize}
	\end{algorithm}

Step 2 can be analyzed directly on the exponents $u$ of the terms with sign $-\tau$: restrict  $u$ to the components $n+1,\dots,n+m$ and choose the corresponding sets $C$ yielding to vectors with maximal support.

	\begin{example}\label{Example:Applying_Algorithm_on_2_site_phosphorylation}
		Consider the following complete binomial networks:
				\begin{align*}
\cN_1\colon & \begin{array}{l}
		X_0+E\ce{->}X_1+E\ce{->}X_2+E\\
		X_2+F\ce{->}X_1+F\ce{->}X_0+F,
		\end{array}
\\ 
\cN_2  \colon & \begin{array}{ll}
		S_0+E\ce{->}S_1+E & \qquad S_1+F\ce{->}S_0+F\\
		P_0+E\ce{->}P_1+E & \qquad P_1+F\ce{->}P_0+F,
		\end{array}
\\ 
\cN_3  \colon & \begin{array}{ll}
		S_0+E\ce{->}S_1+E & \qquad  S_1+F\ce{->}S_0+F\\
		P_0+S_1\ce{->}P_1+S_1 & \qquad P_1+F\ce{->}P_0+F.
		\end{array}
\end{align*}
We  order the  species of $\cN_1$ as  $X_0,X_1,X_2,E,F$ and for $\cN_2,\cN_3$ we consider $P_0,P_1,S_0,S_1,E,F$. Complexes are ordered as they appear in the reaction network from left to right and from up to down.
For suitable choices of admissible binomial bases, 
we obtain the following exponent matrices $M_1,M_2,M_3$ for the three networks respectively, c.f.~\eqref{Equation Gamma and M}:
{\small \[ \arraycolsep=3pt\def\arraystretch{1} M_1^T = 	\left[\begin{array}{rrccc}
		-1 & 1 & 0 & -1 & 1    \\
		0 & -1 & 1 & -1 & 1 
			\end{array}\right],  \quad M_2^T =\left[\begin{array}{rcrccc}
			0 & 0 & -1 & 1 & -1 & 1  \\
		-1 &1 & 0 & 0 & -1 & 1
				\end{array}\right], \quad  M_3^T =\left[\begin{array}{rcrrrc}
			0 & 0 & -1 & 1 & -1 & 1  \\
		-1 &1 & 0 & -1 & 0 & 1
				\end{array}\right], \]}and we choose  the following matrices of conservation laws:
{\small \[ Z_1 = \left[\begin{array}{ccccc}
		1 & 1 & 1 & 0 & 0   \\
		0 & 0 & 0 & 1 & 0   \\
		0 & 0 & 0 & 0 & 1 
		\end{array}\right],\qquad Z_2 =Z_3= \left[\begin{array}{cccccc}
		1 & 1 & 0 & 0 & 0 & 0 \\
		0 & 0 & 1 & 1 & 0 & 0 \\
		0 & 0 & 0 & 0 & 1 & 0 \\
		0 & 0 & 0 & 0 & 0 & 1 
		\end{array}\right].
		\]}Using this data, we construct  the  matrices $\widetilde{\Gamma}_1,\widetilde{\Gamma}_2,\widetilde{\Gamma}_3$ defined as in \eqref{eq:Gamma_tilde} for the largest canonical extensions associated with $\cN_1,\cN_2,\cN_3$ and find their determinants $D_{\cN_1} ,D_{\cN_2}$ and $D_{\cN_3}$.
		We have $\tau=1$ for all three cases. The sets of monomials with negative coefficients in the corresponding   determinants $D_{\cN_1} ,D_{\cN_2}$ and $D_{\cN_3}$ are respectively
		\begin{align*}
		A_1 =& \big\{
		\lambda_{\{1,2,3,4,7,8,10,11\}},\; \lambda_{\{1,2,3,5,7,8,10,11\}}, \; \lambda_{\{1,2,4,5,7,8,10,11\}}, \; \lambda_{\{1,2,4,7,8,9,10,11\}}, \; \lambda_{\{1,3,4,5,7,8,10,11\}},  \\
		&
		\lambda_{\{2,3,4,5,7,8,10,11\}}, \; \lambda_{\{2,3,5,6,7,8,10,11\}} \big\}, \\ 
		A_2 =& \big\{ \lambda_{\{1,4,5,6,7,8,10,12,13,14\}}, \;
		\lambda_{\{2,3,5,6,8,9,10,11,12,14\}} \big\},\\
		A_3 =& \big\{
		\lambda_{\{1,3,4,5,6,7,9,12,13,14\}}, \;
		\lambda_{\{1,3,4,5,6,7,10,12,13,14\}}, \;
		\lambda_{\{1,3,4,5,6,8,9,12,13,14\}}, \;
		\lambda_{\{1,3,4,5,6,7,9,11,12,14\}}, \\
		&
		\lambda_{\{1,3,4,5,6,7,10,11,12,14\}}, \;
		\lambda_{\{1,3,4,5,6,8,9,11,12,14\}}, \;
		\lambda_{\{1,3,4,5,7,9,10,11,12,14\}}, \;
		\lambda_{\{1,3,4,5,8,9,10,11,12,14\}}, \\
		&
		\lambda_{\{1,3,4,6,7,8,9,12,13,14\}}, \;
		\lambda_{\{1,3,4,6,7,8,10,12,13,14\}}, \;
		\lambda_{\{1,3,4,6,7,8,9,11,12,14\}}, \;
		\lambda_{\{1,3,4,6,7,8,10,11,12,14\}}, \\
		&
		\lambda_{\{1,3,4,6,8,9,10,11,12,14\}}, \;
		\lambda_{\{2,4,6,7,8,9,10,11,12,14\}}
		\big\}.
		\end{align*} 
According to the algorithm, the monomial $ \lambda_{\{1,2,3,4,7,8,10,11\}}$ in $A_1$ gives rise to the set $C_1=\{c_6,c_9\}$, while the monomial $ \lambda_{\{1,2,4,7,8,9,10,11\}}$ yields $C_2=\{c_6\}$. Thus $C_2$ belongs to $\Circuits_1$ while $C_1$ does not. 
Proceeding in this way for all monomials, we obtain
		\begin{align*}
		\Circuits_1 =& \{\{c_6\},\{c_9\}\} = \{\{X_0+E\},\{X_2+F\}\},\\
		\Circuits_2 =& \{\{c_1,c_7\},\{c_3,c_5\}\}	= \{\{S_0+E,P_1+F\},\{S_1+F,P_0+E\}\},\\
		\Circuits_3 =& \{\{c_7\},\{c_3,c_5\},\{c_4,c_5\}\}= \{\{P_1+F\},\{S_1+F,P_0+S_1\},\{S_0+F,P_0+S_1\}\}.
		\end{align*}
We conclude that  motifs (g), (i) and (k) in \cite{Feliu-EnzymeSharing} are multistationary, since they are 
extensions of  $\cN_1,\cN_2,\cN_3$ respectively, which satisfy the generalized realization condition and the set of inputs 
of their intermediates belong to the respective multistationarity structures.

As illustrated by these three examples, the elements of the set of circuits might not have the same cardinality.
	\end{example}

	\begin{remark}\label{rk:approaches}
Algorithm \ref{Algorith:Circuits_Main} provides a direct way to detect the sets of complexes that contribute to multistationarity. The method is appealing because, for small networks, the multistationarity structure can be found by simple visual inspection of one multivariate polynomial. 
The brute force alternative strategy for finding the multistationarity structure consists in 
searching for the circuits by computing $\det(\widetilde{\Gamma}_C)$ for several canonical extensions. 
One starts from one of the smallest subsets $C$ of $\cC$ and computes $\det(\widetilde{\Gamma}_C)$. If this polynomial has terms with different sign or is zero, then  we add $C$ to $\Circuits$, and remove $C$ and all its supersets from $\mathcal{P}(\cC)$ before proceeding in the same way with the next smallest set. Alternatively, one can start the search with one of the largest subsets $C$ of $\cC$ and compute $\det(\widetilde{\Gamma}_C)$. If the determinant does not have terms with different sign, then we  remove all subsets of $C$ from $\mathcal{P}(\cC)$. If it has terms of both signs or is zero, then we check the subsets of $C$ with one less element. If none of them is multistationary, then we add $C$ to $\Circuits$ and remove all its subsets from the search.

Going from small to large sets has the advantage of involving the computation of smaller determinants. Our algorithm requires the computation of only one  determinant, but it can be large. So, for large networks, it might be advantageous to adopt the search approach starting with small sets described here.
\end{remark}

	\subsection{$\e$-site phosphorylation network}\label{sec:nsite}
	In this section we find the multistationarity structure of the $\e$-site distributive sequential phosphorylation network given as follows (see e.g.~\cite{Dickenstein-Toric,Wang:2008dc}): 
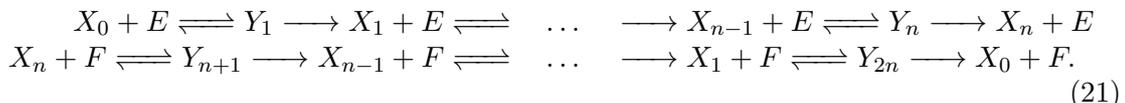
\begin{equation}\label{eq:nsite}
\begin{array}{rcl}
	X_0+E\ce{<=>}Y_1\ce{->}X_1+E\ce{<=>} & \dots &  \ce{->}X_{\e-1}+E\ce{<=>}Y_\e\ce{->}X_\e+E\\
	X_\e+F\ce{<=>}Y_{\e+1}\ce{->}X_{\e-1}+F\ce{<=>} & \dots & \ce{->}X_1+F\ce{<=>}Y_{2\e}\ce{->}X_0+F.
	\end{array}
\end{equation}
By removing the intermediates $Y_1,\dots,Y_{2\e}$, the core network associated with the $\e$-site phosphorylation network is
	\begin{equation}\label{eq:core_nsite} 
	\cN\colon\begin{array}{l}
	X_0+E\ce{->[k_1]}X_1+E\ce{->[k_2]}\dots\ce{->[k_{\e-1}]} X_{\e-1}+E\ce{->[k_{\e}]} X_{\e}+E\\
	X_\e+F\ce{->[k_{\e+1}]} X_{\e-1}+F\ce{->[k_{\e+2}]} \dots\ce{->[k_{2\e-1}]} X_1+F\ce{->[k_{2\e}]} X_0+F.
	\end{array}
	\end{equation}
	Since $(1,\dots,1)$ is in the kernel of the stoichiometric matrix of $\cN$, (surj) holds. Further, the rank of $\cN$ is $\e$ and an admissible binomial basis of the steady state ideal is
	\[B:=\big\{-k_1x_0e+k_{2\e}x_1f,\dots,-k_{\e}x_{\e-1}e+k_{\e+1}x_\e f \big\}.\]
\blue{This basis can be easily obtained by performing linear combinations of the steady state equations. It has been used in several works such as \cite{Dickenstein-Toric,Dickenstein-Bihan-Magali}.}
We conclude that (rank) also holds and $\cN$ is a complete binomial network. 
By Proposition \ref{Lemma:Realization_more_classes} (ii),  the generalized realization condition holds  for the $\e$-site distributive sequential phosphorylation networks given in \eqref{eq:nsite}.

We order the set of species as
$X_0,\,X_1,\,\dots,\,X_\e,\,E, F$, 
and denote the complexes of the core network as 
	\[
	  c_1=X_0+E,\quad \dots \quad c_{\e+1}=X_\e+E,\quad
	  c_{\e+2}=X_\e+F, \quad \dots \quad c_{2\e+2}=X_0+F.
	\]
The largest canonical network consists of the reactions of $\cN$ together with  the reactions
$c_i \ce{<=>} Y_{i}.$
The matrix $\blue{(M^T)_\lambda}\in\R^{\e\times(\e+3)}$  associated with $B$ and a choice of  $Z\in\R^{3\times(\e+3)}$ are
	\begin{align*}
	\blue{(M^T)_\lambda} =& \left[\begin{array}{ccccccc}
	-\lambda_1 & \lambda_2 & & & 0 & -\lambda_{\e+2} & \lambda_{\e+3} \\
	& -\lambda_2 & \lambda_3 & & & -\lambda_{\e+2} & \lambda_{\e+3}\\
	& & \ddots & \ddots & & \vdots & \vdots\\
	0 & & & -\lambda_\e & \lambda_{\e+1} & -\lambda_{\e+2} & \lambda_{\e+3}
	\end{array}\right], \quad 
	Z =& \left[\begin{array}{ccccc}
	1  & \cdots & 1 & 0 & 0\\
	0  & \dots &  0 & 1 & 0\\
	 0 & \dots &  0 & 0 & 1
	\end{array}\right].
	\end{align*}

\begin{proposition}\label{prop:nsite_canonical}
For the $\e$-site phosphorylation network with $\e\geq 2$, we have 
\[\Circuits=\{\{c_i\}\mid i\neq \e,\e+1,2\e+1,2\e+2\} = \big\{X_0+E,\dots,X_{\e-2}+E, \; X_\e+F,\dots,X_2+F \big\} .\]
If $\e=1$, then $\Circuits=\emptyset$, since the largest canonical extension  is not multistationary.
\end{proposition}
\begin{proof}
The case $\e=1$ follows by computing the determinant of the largest canonical extension and checking that it is non-zero and that all coefficients have the same sign. 

Hence assume that $\e\geq 2$.
It is enough to first show that $\{c_i\}\in \Circuits $ if $i\neq \e,\e+1,2\e+1,2\e+2$ and then that  $\{c_{\e},c_{\e+1},c_{2\e+1},c_{2\e+2}\}\not\in\Mult$. 

So let $i\in[\e+1]$ and define
		{\small \[\Omega(i)=\left[\begin{array}{rrrrrrrrr;{2pt/2pt}r}
		-1 & 1 & & & & & 0 & -1 & 1 & 0 \\
		& \ddots & \ddots & & & & & \vdots & \vdots & \vdots \\
		& & -1 & 1 & & & & -1 & 1 & 0 \\
		& & & -1 & 1 & & & -1 & 1 & 0 \\
		& & & & \ddots & \ddots & & \vdots & \vdots & \vdots \\
		0 & & & & & -1 & 1 & -1 & 1 & 0 \\
		\hdashline[2pt/2pt]
		0 & \cdots & 0 & -1 & 0 & \cdots & 0 & -1 & 0 & 1
		\end{array}\right]\in\R^{(\e+1)\times(\e+4)},\]}where the $-1$ in the last row  is in the $i$-th column. Then  we have that 
\[ \widetilde{\Gamma}_{\{c_i\}}=\left[\begin{array}{c}
	\Omega(i)\diag(\lambda_1,\dots,\lambda_{\e+4}) \\
	\hdashline[2pt/2pt]
\begin{array}{ccc;{2pt/2pt}c}
 & & & 1 \\
 \qquad & Z & \qquad & 0 \\
 & & & 0
\end{array}
	\end{array}\right]\in\R^{(\e+4)\times(\e+4)}.\] 
	For $J\subseteq [\e+4]$ of cardinality $3$, we denote by $\Omega(i)_J$ the $(\e+1)\times (\e+1)$ submatrix of $\Omega(i)$ obtained by deleting the columns with index in $J$. 
We expand the determinant of $\widetilde{\Gamma}_{\{c_i\}}$ along the last three rows and obtain
	\begin{align*}
			\det\big(\widetilde{\Gamma}_{\{c_i\}}\big) =& \sum_{j=1}^{\e+1}(-1)^{5\e+14+j}\begin{vmatrix}
			1 & 0 & 0\\
			0 & 1 & 0\\
			0 & 0 & 1
			\end{vmatrix}\det \hspace{-0.07cm}\big(\Omega(i)_{\{j,\e+2,\e+3\}}\big)\lambda_{[\e+4]\setminus\{j,\e+2,\e+3\}}+\\
			& \sum_{j=1}^{\e+1}(-1)^{5\e+16+j}\begin{vmatrix}
			1 & 0 & 1\\
			0 & 0 & 1\\
			0 & 1 & 0
			\end{vmatrix}\det \hspace{-0.07cm}\big(\Omega(i)_{\{j,\e+3,\e+4\}}\big)\lambda_{[\e+4]\setminus\{j,\e+3,\e+4\}}+\\
			& (-1)^{6\e+18}\begin{vmatrix}
			0 & 0 & 1\\
			1 & 0 & 1\\
			0 & 1 & 0
			\end{vmatrix}\det \hspace{-0.07cm}\big(\Omega(i)_{\{\e+2,\e+3,\e+4\}}\big)\lambda_{[\e+4]\setminus\{\e+2,\e+3,\e+4\}}\\
			=& \sum_{j=1}^{\e+1}(-1)^{\e+j}\det \hspace{-0.07cm}\big(\Omega(i)_{\{j,\e+2,\e+3\}}\big)\lambda_{[\e+4]\setminus\{j,\e+2,\e+3\}}+\\
			& \sum_{j=1}^{\e+1}(-1)^{\e+j+1}\det \hspace{-0.07cm}\big(\Omega(i)_{\{j,\e+3,\e+4\}}\big)\lambda_{[\e+4]\setminus\{j,\e+3,\e+4\}}+\\
			& \det \hspace{-0.07cm}\big(\Omega(i)_{\{\e+2,\e+3,\e+4\}}\big)\lambda_{[\e+4]\setminus\{\e+2,\e+3,\e+4\}}.
			\end{align*}
We see from this expansion that the coefficient of $\lambda_{[\e+4]\setminus\{1,\e+2,\e+3\}}$ 
is $(-1)^{\e+1}\det\hspace{-0.07cm}\big(\Omega(i)_{\{1,\e+2,\e+3\}}\big)$. We have that
{\small 
\begin{equation}\label{eq:subomega}
\Omega(i)_{\{1,\e+2,\e+3\}}=\left[\begin{array}{rrrrr;{2pt/2pt}r}
		1 & & & &  0 & 0  \\
		-1 & \ddots  & & & & \\
		& \ddots & \ddots  & & &  \vdots\\
		& &  \ddots & \ddots & &  \\
		0 & & &  -1 & 1 & 0 \\ \hdashline[2pt/2pt]
		0 & \cdots  & -1  & \cdots & 0 & 1
		\end{array}\right]\in\R^{(\e+1)\times(\e+1)},
		\end{equation}}where the $-1$ in the last row is in position $i-1$ if $i>1$ and is not there  if $i=1$. 
		Clearly, $(-1)^{\e+1}\det\hspace{-0.07cm}\big(\Omega(i)_{\{1,\e+2,\e+3\}}\big)=(-1)^{\e+1}$.
		
Consider now the coefficient of  $\lambda_{[\e+4]\setminus\{\e+1,\e+3,\e+4\}}$, which is $(-1)^{\e+\e+1+1}\det\hspace{-0.07cm}\big(\Omega(i)_{\{\e+1,\e+3,\e+4\}}\big)$. We have that 
{\small \[\Omega(i)_{\{\e+1,\e+3,\e+4\}}=\left[\begin{array}{rrrrr;{2pt/2pt}r}
		-1 & 1 & & & 0 & -1   \\
		& -1 & \ddots & & \\
		& & \ddots & \ddots & & \vdots\\
		& & & \ddots & 1 &  \\
		0 & & & & -1 & -1 \\ \hdashline[2pt/2pt]
		0 & \cdots  & -1 & \cdots & 0 & 1
		\end{array}\right]\in\R^{(\e+1)\times(\e+1)}, \]}where the $-1$ in the  last row is in position $i$ if $i\leq \e$, and there is no $-1$ if $i=\e+1$.
Replacing the last row of  $\Omega(i)_{\{\e+1,\e+3,\e+4\}}$ with minus the sum of the rows from $i$ to $\e$,
we obtain the matrix
		{\small	\[\left[\begin{array}{crcr;{2pt/2pt}c}
			-1 & 1 & & 0 & -1 \\
			& -1 & \ddots & & \vdots \\
			& & \ddots & 1 & -1 \\
			0 & & & -1 & -1 \\
			\hdashline[2pt/2pt]
			& & 0 & & \e-i
			\end{array}\right]\in\R^{(\e+1)\times(\e+1)}.\]}It follows that  the coefficient of  $\lambda_{[\e+4]\setminus\{\e+1,\e+3,\e+4\}}$ is
$(-1)^{\e}(\e-i)$. 

This shows that if $i<\e$, then the coefficients of  $\lambda_{[\e+4]\setminus\{\e+1,\e+3,\e+4\}}$ and 
		$\lambda_{[\e+4]\setminus\{1,\e+2,\e+3\}}$  have opposite non-zero signs, and hence $\{c_i\}$ is a circuit.
		For $\e+1<i\leq 2\e$ the claim  follows by the symmetry of the network after interchanging $E$ and $F$ and sending $X_0,\dots,X_\e$ to $X_\e,\dots,X_0$.

\medskip
All that remains is to show that  $C=\{c_{\e},c_{\e+1},c_{2\e+1},c_{2\e+2}\}\not\in\Mult$. The matrix $\widetilde{\Gamma}_{C}\in\R^{(\e+7)\times(\e+7)}$ is:
{\small  \[\left[\begin{array}{cccccccc;{2pt/2pt}cccc}
	-\lambda_1 & \lambda_2 &  & & & 0 & -\lambda_{\e+2} & \lambda_{\e+3} & & & & \\
	& -\lambda_2 & \lambda_3  & & & &  &  & & & & \\
 	& & \ddots &  \ddots&  & \vdots & \vdots &\vdots & &   \multicolumn{2}{c}{ 0}  & \\
 	& & &  \ddots &   \lambda_{\e}  & 0  & -\lambda_{\e+2}  & \lambda_{\e+3} & & & & \\
	0 & & &  & -\lambda_{\e} & \lambda_{\e+1} & -\lambda_{\e+2} & \lambda_{\e+3} & & & & \\
	\hdashline[2pt/2pt]
	0 &0 & \cdots & \cdots  &  -\lambda_\e & 0 & -\lambda_{\e+2} & 0 & \lambda_a & & & 0 \\
	0 & 0  & \cdots & \cdots  &  0 & -\lambda_{\e+1} & -\lambda_{\e+2} & 0 & & \lambda_b & & \\
	0 & -\lambda_2  & \cdots & \cdots  & 0   & 0 & 0 & -\lambda_{\e+3}  & & &  \lambda_c & \\
	-\lambda_1 & 0 & \cdots  & \cdots  & 0  & 0 & 0 & -\lambda_{\e+3} & 0 & & & \lambda_d \\
	\hdashline[2pt/2pt]
	1 & 1 & \cdots   & \cdots &  1 & 1 & 0 & 0 & 1 & 1 & 1 & 1 \\
      0 	& 0  &\cdots  & \cdots  & 0  & 0 & 1 & 0 & 1 & 1 & 0 & 0 \\
	0 	& 0  &\cdots  & \cdots  & 0  & 0 & 0 & 1 & 0 & 0 & 1 & 1 
	\end{array}\right].\]}By performing row operations, we transform $\widetilde{\Gamma}_{C}$ into a block triangular matrix with diagonal blocks of size $\e$ and $7$,  respectively,  as follows:
\begin{itemize}
\item Subtract  the sum of the rows $1,\dots,\e$ from the $(\e+4)$-th row.
\item Subtract the sum of the rows $2,\dots,\e$ from the $(\e+3)$-th row.
\item Subtract  the $\e$-th row from the $(\e+1)$-th row.
\item Add to the $(\e+5)$-th row the following linear combination of the first $\e$ rows:
\[ \sum_{i=1}^\e \Big( \tfrac{1}{\lambda_1} + \dots + \tfrac{1}{\lambda_i} \Big) f_i,\]
where $f_i$ is the $i$-th row.
\end{itemize}
After these operations, which preserve the determinant, we obtain the following matrix:
{\small \[\left[\begin{array}{cccc;{2pt/2pt}ccc;{2pt/2pt}cccc}
	-\lambda_1 & \lambda_2 &  &   0 & 0 & -\lambda_{\e+2} & \lambda_{\e+3} & & & & \\
 	&  \ddots &  \ddots&  & \vdots & \vdots &\vdots & & \multicolumn{2}{c}{ 0} & \\
 	& &  \ddots &   \lambda_{\e}  &   0  & -\lambda_{\e+2}  & \lambda_{\e+3}  & & & & \\
	0 & &  & -\lambda_{\e} & \lambda_{\e+1} & -\lambda_{\e+2} & \lambda_{\e+3} & & & & \\
	\hdashline[2pt/2pt]
	0  & \cdots & \cdots  &  0 & -\lambda_{\e+1} &0 &  -\lambda_{\e+3} & \lambda_{\e+4} & & & 0 \\
	0   & \cdots & \cdots  &  0 & -\lambda_{\e+1} & -\lambda_{\e+2} & 0 & & \lambda_{\e+5} & & \\
	0  & \cdots & \cdots  & 0   & -\lambda_{\e+1} &  (\e-1) \lambda_{\e+2}  & -\e \lambda_{\e+3}  & & &  \lambda_{\e+6} & \\
	0  & \cdots  & \cdots  & 0  & -\lambda_{\e+1} & \e  \lambda_{\e+2}  & -(\e+1)\lambda_{\e+3} & 0 & & & \lambda_{\e+7} \\
	\hdashline[2pt/2pt]
	0  & \cdots   & \cdots &  0 & 1+z_1\lambda_{\e+1} & -z_2 \lambda_{\e+2} & z_2 \lambda_{\e+3} & 1 & 1 & 1 & 1 \\
      0 	 &\cdots  & \cdots  & 0  & 0 & 1 & 0 & 1 & 1 & 0 & 0 \\
	0   &\cdots  & \cdots  & 0  & 0 & 0 & 1 & 0 & 0 & 1 & 1 
	\end{array}\right],\]}where
\[ z_1 = \sum_{i=1}^\e \tfrac{1}{\lambda_i},\qquad z_2 = \sum_{i=1}^\e \tfrac{\e - i + 1}{\lambda_i}.\]
The determinant of $\widetilde{\Gamma}_{C}$ is therefore equal to $(-1)^\e \lambda_{[\e]}$ times the determinant of the inferior diagonal block of size $7\times 7$ of the matrix above. 
We compute this determinant and obtain  the following expression:
\begin{align*}
& \lambda_{\e+2}\lambda_{\e+4}\lambda_{\e+6}\lambda_{\e+7}(1+ ( z_{1}+z_{2} ) \lambda_{\e+1}) 
+\lambda_{\e+1}\lambda_{\e+4}(\lambda_{\e+5}\lambda_{\e+7}+\lambda_{\e+6}\lambda_{\e+7}
+\lambda_{\e+5}\lambda_{\e+6})
\\ &  +\lambda_{\e+1}\lambda_{\e+5}\lambda_{\e+6}\lambda_{\e+7}(1+z_{1}\lambda_{\e+4}  +z_{2}\lambda_{\e+2})  + \big( ( \e\,z_{1}-z_{2} ) \lambda_{\e+1}    +\e\big )\lambda_{\e+3}\lambda_{\e+4}\lambda_{\e+5}\lambda_{\e+7}
\\ & + \big(( \e\,z_{1}+z_{1}-z_{2} ) \lambda_{\e+1}  + \e+1 \big)\lambda_{\e+3}\lambda_{\e+4}\lambda_{\e+5}\lambda_{\e+6}+\lambda_{\e+4}\lambda_{\e+5}\lambda_{\e+6}\lambda_{\e+7}
\\ &  +  
 \big(z_{1}\lambda_{\e+1}\lambda_{\e+2}\lambda_{\e+3}+\lambda_{\e+1}\lambda_{\e+2} +\lambda_{\e+1}\lambda_{\e+3}  +\lambda_{\e+2}\lambda_{\e+3} \big) \big( ( \e+1 ) \lambda_{\e+4}\lambda_{\e+6}+\e\, \lambda_{\e+4}\lambda_{\e+7} \\ &\qquad  +\e\,\lambda_{\e+5}\lambda_{\e+6}+( \e-1 ) \lambda_{\e+5}\lambda_{\e+7} \big).
\end{align*}
Since  $\e z_1 \geq z_2$ and $\e \geq 2$, this determinant is strictly positive. Hence, the determinant of $\widetilde{\Gamma}_{C}$ has sign $(-1)^\e$. By Theorem \ref{Theorem:Multistationarity_and_Determinant_Condition}, we conclude that $\{c_{\e},c_{\e+1},c_{2\e+1},c_{2\e+2}\}\not\in\Mult$.
\end{proof}

In view of Proposition \ref{prop:nsite_canonical} and  Theorem~\ref{Theorem Multistationartity and Realization}(ii) 
we obtain the following theorem.

	\begin{thm}\label{Theorem:Circuits_of_n_site_phosphorylation}
	Let $\widetilde{\cN}$ be an extension of the core $\e$-site phosphorylation network in \eqref{eq:core_nsite}  via the addition of intermediates that satisfies the generalized realization condition. Then $\widetilde{\cN}$ is multistationary if and only if at least one of $X_0+E,\dots,X_{\e-2}+E,X_\e+F,\dots,X_2+F$ is an input of an intermediate.
	\end{thm}

Note that the network in Example~\ref{Example:Canonical_Representatives_MPAK} is an extension of the $2$-site phosphorylation network, with set of inputs $C=\{X_0+E,X_1+E,X_2+F,X_1+F,X_0+F\}$. This network satisfies the generalized realization condition by Example \ref{Example:Generalized_Realization_Condition_MAPK_Nsite}. By Theorem~\ref{Theorem:Circuits_of_n_site_phosphorylation}, we conclude that the network is multistationary.

For the $\e$-site phosphorylation network for a fixed $\e$, Algorithm \ref{Algorith:Circuits_Main} requires the computation of one large determinant. The search approach described in Remark~\ref{rk:approaches}, 
stops after computing $2\e+14$ determinants, if we start with the small subsets, while it stops after computing
	\[\sum_{i=1}^3\binom{2\e-2}{i}+\sum_{i=4}^{2\e+2}\binom{2\e+2}{i}\]
determinants if we start with large subsets. 
For example, if $\e=2,3$, the first approach requires the computation of $18$ and $20$ determinants, and the second approach requires the computation of  $25$ and $177$ determinants respectively. In these cases, the computation of the determinants takes negligible time, and therefore our algorithm is the fastest strategy.

	\section{Realization conditions}\label{sec:Realization_conditions}
	
\blue{	In this section we briefly discuss generic algebraic approaches to decide whether the realization conditions are satisfied. 
We proceed to explain how we can break the problem of checking whether the realization conditions are satisfied, into checking the conditions  for a collection of (smaller) subnetworks. We conclude with a list of small networks that satisfy them. These small networks cover typical cases arising in applications. 
}

The two realization conditions concern the surjectivity of a rational map on the positive orthant. 
Specifically, let $\widetilde{\cN}$ be an extension of $\cN$ via the addition of the intermediates $Y_1,\dots,Y_m$ and $C'$ be the set of input complexes that do not belong to the stoichiometric subspace. Consider the following maps from $\R_{>0}^{\widetilde{r}}$:
\begin{align}\label{eq:phi}
\blue{\phi^*(\k)}  &= \big(\phi_{c\rightarrow c'}(\k)\mid c\rightarrow c'\in\cR\big) \in    \R_{>0}^{r}, \\
\phi'(\k) &= \Big( \blue{\phi^*(\k)},\big( \sum\nolimits_{i\in [m]} \mu_{i,c}(\k)\mid c\in C'\big)\Big) \in  \R_{>0}^{r}\times\R_{>0}^{C'}.\label{eq:phi:prime}
\end{align}
The generalized realization condition is equivalent to the surjectivity of $\phi'$ and the realization condition to the surjectivity of \blue{$\phi^*$}. So let  $f=(\tfrac{f_1}{g_1},\dots,\tfrac{f_m}{g_m})$ be an arbitrary map from $\R_{>0}^n$ to $\R_{>0}^m$, defined by rational functions $\tfrac{f_i}{g_i}\in \R(x_1,\dots,x_n)$. Consider the ideal $I=\big\langle g_1y_1-f_1,\dots,g_my_m-f_m,1-z\prod_{i=1}^mg_i\big\rangle\subseteq\R[y_1,\dots,y_m,x_1,\dots,x_n,z]$. 
As discussed in \S 2.3 of the electronic supplementary material of \cite{Feliu-Simplifying}, if 
 $I\cap\R[y_1,\dots,y_m]\neq \{0\}$, then $f$ is not surjective, but the reverse does not necessarily hold.
	
Another approach is to use Cylindrical Algebraic Decomposition (CAD) \cite{basu2007algorithms,SolvingParametricPolynomialSystemsFabrice,MaplePackageCAD}. 
Consider the parametric multivariate system of equations $(x_1,\dots,x_n,z)\in V(I)$ with $y_1,\dots,y_m$ treated as parameters and all variables and parameters constrained to be real and positive. The map $f$ is surjective if and only if 
this system has at least one positive real solution when evaluated at the sample parameter point of all cells obtained after performing CAD.
This approach fully characterizes whether $f$ is surjective, but CAD is computationally expensive. In particular, the number of cells is doubly exponential in  the number of variables and parameters, and depends also on the degree and number of polynomials in the system  \cite[Theorem 5]{NumberOfCellsCAD}. Therefore the use of CAD is impractical already in relatively small examples.

	\begin{example}\label{Example:CAD_vs_Elimination_Realization}
		We consider the following core network and its extension via the addition of one intermediate $Y$:
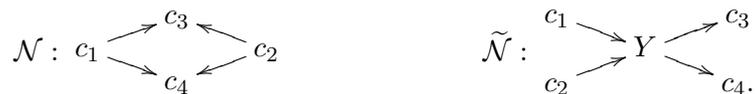
		\begin{align*}
		\cN& :\adjustbox{valign=c}{\xymatrix @C=1.5pc @R=0pc { & c_3 & \\
				c_1\ar[ur]\ar[dr] &  & c_2\ar[ul]\ar[dl]\\
				& c_4 & }} & 
		\widetilde{\cN} & :\adjustbox{valign=c}{\xymatrix@C=1.5pc @R=0pc { c_1\ar[dr] & & c_3\\
				& Y\ar[ur]\ar[dr] &  \\
				c_2\ar[ur] & & c_4.}}
		\end{align*}
Using CAD on the system of equations describing the realization condition, we obtain three cells. The sample point of each cell yields a system with infinitely many positive solutions. Therefore the realization condition holds.
	\end{example}

In view of the difficulties of checking the realization conditions in practice, 
we start by understanding how the coefficients $\mu_{i,c}$ are found. 
Let $\widetilde{\cN}$ be an extension of $\cN$ via the addition of  intermediates $Y_1,\dots,Y_m$. 
Consider the digraph associated with  $\widetilde{\cN}$ and let $\mathcal{Y}_1,\dots,\mathcal{Y}_{t'}$ denote the vertex sets of the connected components of the subgraph induced by the subset of vertices  $\{Y_1,\dots,Y_m\}$. 
For each non-intermediate complex $c$ and intermediate $Y_i$, consider the labeled digraph $G_{i,c}$ with vertex set 
$\mathcal{Y}_\ell \cup \{\star\}$ if $Y_i\in \mathcal{Y}_\ell$. 
Labeled edges are $Y_i\ce{->[\k_{Y_i\rightarrow Y_j}]}Y_j$ if $Y_i\rightarrow Y_j\in\widetilde{\cR}$, $\star\ce{->[\k_{c\rightarrow Y_i}]}Y_i$ if $c\rightarrow Y_i\in\widetilde{\cR}$ and $Y_i\ce{->[\beta_i]}\star$ with $\beta_i=\sum_{Y_i\rightarrow c'}\k_{Y_i\rightarrow c'}$ if $\beta_i\neq 0$. For each vertex $v$ of $G_{i,c}$, define $\Theta_{i,c}(v)$ to be the set of all spanning trees rooted at $v$, that is,  $v$ is the only vertex with zero outdegree. Given a  tree $\tau$, let $\pi(\tau)$ be the product of all labels of the edges of $\tau$. Then
\begin{equation}\label{eq:muic}
\mu_{i,c}=\tfrac{\sum_{\tau\in\Theta_{i,c}(Y_i)}\pi(\tau)}{\sum_{\tau\in\Theta_{i,c}(\star)}\pi(\tau)}.
\end{equation}
The numerator of $\mu_{i,c}$ is linear in the reaction rate constants of the form $\kappa_{c\rightarrow Y_j}$,  and these reaction rate constants do not appear in the denominator. 
To read more about properties of the $\mu_{i,c}$'s and how to compute them using the Matrix-Tree theorem, see  \cite{Feliu-Simplifying}.
	
The components of \blue{$\phi^*$} and $\phi'$ might not  depend on the reaction rate constants of all reactions in  the network. 
Specifically,  from \eqref{eq:muic} and \eqref{Equation:phi} it follows that $\phi_{c\rightarrow c'}$ depends on $c\rightarrow c'$, if this reaction belongs to $\widetilde{\cR}$, and possibly on the reactions involving intermediates in the sets $\mathcal{Y}_j$ such that there exists a path from $c$ to $c'$ with all intermediates in $\mathcal{Y}_j$.
So for each reaction, we consider the union of these relevant sets of intermediates $\mathcal{Y}_j$. 
Then $\phi_{c_1\rightarrow c_1'}$ and $\phi_{c_2\rightarrow c_2'}$ do not depend on a common reaction rate constant if the sets of intermediates corresponding to $c_1\rightarrow c_1'$ and $c_2\rightarrow c_2'$ are disjoint. In this way we  partition  $\widetilde{\cR}$ into subsets of reactions, that is, subnetworks, for which surjectivity of the map $\phi$ can be checked independently on each smaller network. 

We proceed similarly for $\phi'$, but in this case the relevant sets of intermediates $\mathcal{Y}_j$ are those for which there exists a path from $c$ to at least one $Y_{i}\in \mathcal{Y}_j$ (or equivalently, $\mu_{i,c}\neq 0$).

\begin{example}\label{Example:Generalized_Realization_Condition_MAPK_Nsite}
We consider the generalized realization condition for Example~\ref{Example:Canonical_Representatives_MPAK}. By the discussion above, this condition needs to be checked independently on the following three subnetworks: 
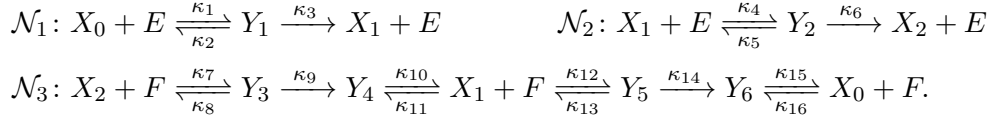
\[\begin{aligned}
	\cN_1\colon	X_0+E & \ce{<=>[\kappa_1][\kappa_2]} Y_1\ce{->[\kappa_3]} X_1+E \qquad \qquad  	\cN_2\colon X_1+E\ce{<=>[\kappa_4][\kappa_5]} Y_2\ce{->[\kappa_6]} X_2+E\\
			\cN_3\colon X_2+F & \ce{<=>[\kappa_7][\kappa_8]} Y_3\ce{->[\kappa_9]} Y_4\ce{<=>[\kappa_{10}][\kappa_{11}]} X_1+F\ce{<=>[\kappa_{12}][\kappa_{13}]} Y_5\ce{->[\kappa_{14}]} Y_6\ce{<=>[\kappa_{15}][\kappa_{16}]} X_0+F.	
	\end{aligned}
\]
For the three subnetworks the generalized realization condition holds due to Proposition~\ref{Lemma:Realization_more_classes}(ii) below.
	\end{example}

	We  next show that the realization condition holds for specific classes of intermediates without the need to  do any extra computations.
	
	\begin{proposition}\label{Lemma:Realization_more_classes}
		The realization condition holds for the following types of \blue{extended}  networks via the addition of  intermediates $Y_1,\dots,Y_m$. 	
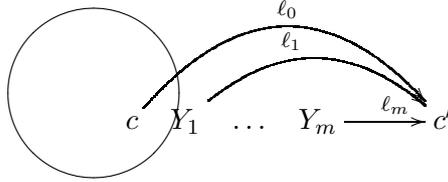
		\begin{enumerate}[(i)]
			\item 
			\[
			\xymatrix @C=0.5pc @R=0pc{
				& & & & & & & \\
				& & & & & & &  \\
				& c\ar@/^3pc/[rrrrrr]^{\ell_0} & Y_1\ar@/^2pc/[rrrrr]^{\mkern-36mu\ell_1}  & \dots & Y_m\ar[rrr]^{\quad\ell_m} & & & c'\\
				& & & & & & &  
				\save "2,1"."4,6"*\frm{e}
				\restore
			}\]
			with an arbitrary digraph structure among the complexes $c,Y_1,\dots,Y_m$ such that there is a path from $c$ to all $Y_i$, and where some reactions with label $\ell_i$ might not exist.
			
			\item $c\ce{<->}Y_1 \ce{<->} Y_2\ce{<->} \dots \ce{<->} Y_m \ce{<->} c'$, provided $\{Y_1,\dots,Y_m\}$ is  a set of intermediates, and where $\ce{<->}$ means the reaction can be either irreversible or reversible. These networks satisfy also  the generalized realization condition. Further,  a union of subnetworks of this form such that the sets of intermediates of each subnetwork do not intersect, satisfies also the generalized realization condition.

			\item
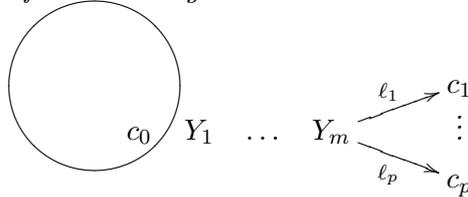
			\[
			\xymatrix @C=0.5pc @R=0pc{
				& & &  & & & & c_1\\
				& c_0 & Y_1  & \dots & Y_m \ar[urrr]^{\ell_1}\ar[drrr]_{\ell_p} & & & \vdots\\
				& & &  & & & & c_p
				\save "1,1"."3,6"*\frm{e}
				\restore
			}\]
			with an arbitrary digraph structure among the complexes $c_0,Y_1,\dots,Y_m$ such that there exists a directed path from $c_0$ to $Y_m$, \blue{and where reactions with label $\ell_1,\dots,\ell_p$ have source $Y_m$.}
		\end{enumerate}
	
	\end{proposition}
	\begin{proof}

		\medskip
		\noindent
		(i) The realization condition is equivalent to the  scalar-valued map $\ell_0+\sum_{i=1}^m\ell_i \mu_{i,c}$ being surjective.
			This map is  linear in $\ell_0,\k_{c\rightarrow Y_1},\dots,\k_{c\rightarrow Y_m}$ (some might be zero, but at least one is non-zero).   Hence the statement is clear. 
		
		\medskip
		\noindent
	(ii) We start with the case with only one such block. 
We write
\begin{equation}\label{eq:reversible}
\widetilde{\cN}\colon \quad c\ce{<=>[\kappa_1][\kappa_2]}Y_1\ce{<=>[\kappa_3][\kappa_4]}Y_2\ce{<=>[\kappa_5][\kappa_6]}\dots\ce{<=>[\kappa_{2m-1}][\kappa_{2m}]}Y_m\ce{<=>[\kappa_{2m+1}][\kappa_{2m+2}]}c'.
\end{equation}
If not all reactions are reversible, then we assume that the reaction of the core network is $c\rightarrow c'$. This means that all reactions with label with odd subindex are present, and the reverse reactions might or might not be present. 

We can assume without loss of generality that 
 neither $c\rightarrow c'$ nor $c'\rightarrow c$ belong to $\widetilde{\cN}$ (if a map is surjective between two positive orthants, adding an extra variable that sums to one component preserves surjectivity).

We have that $\phi_{c\rightarrow c'}(\k)= \k_{2m+1}\mu_{m,c}(\kappa)$ and  $\phi_{c'\rightarrow c}(\k)=\k_2\mu_{1,c'}(\kappa)$ (the latter being zero in the irreversible case).
Throughout we assume that the set $C'$ used to define $\phi'$ equals $\{c,c'\}$. This is the worst case scenario. 

We show by induction on $m$ that this network satisfies the generalized realization condition. For $m=1$, if all reverse reactions are present we have that 
\[ \phi'(\k_1,\k_2,\k_3,\k_4)= \big(\phi_{c\rightarrow c'},\phi_{c'\rightarrow c}, \mu_{1,c}, \mu_{1,c'}\big)= \big( \tfrac{\k_1\k_3}{\k_2+\k_3}, \tfrac{\k_2\k_4}{\k_2+\k_3}, \tfrac{\k_1}{\k_2+\k_3}, \tfrac{\k_4}{\k_2+\k_3}  \big).\]
 A missing reverse reaction corresponds to setting the reaction rate constant equal to zero, and projecting $\phi'$ away from the  components that become zero. 
We confirm using CAD that this map is surjective when restricted to the positive orthants, in the four scenarios obtained by considering none, one or both reverse reactions.

Assume now that \eqref{eq:reversible} satisfies the generalized realization condition for $m-1$. 
We view $\widetilde{\cN}$ as an extended network of 
\[\overline{\cN}\colon \quad c\ce{<=>[\overline{\kappa}_1][\overline{\kappa}_2]} Y_2\ce{<=>[\kappa_5][\kappa_6]}\dots\ce{<=>[\kappa_{2m-1}][\kappa_{2m}]}Y_{m} \ce{<=>[\kappa_{2m+1}][\kappa_{2m+2}]}c',\]
via the addition of one intermediate $Y_1$. If we let $\widetilde{\k} =(\k_{1}, \k_2,\k_3,\k_4)$, this gives rise to the following relevant functions
\begin{align}
\widetilde{\mu}_{1,c} (\widetilde{\k} ) & = \tfrac{\k_{1}}{\k_{2}+\k_{3}},&  \widetilde{\mu}_{1,Y_2}  (\widetilde{\k} ) &= \tfrac{\k_{4}}{\k_{2}+\k_{3}}, &
 \overline{\k}_{1} &= \k_{3} \widetilde{\mu}_{1,c} (\widetilde{\k} ) , & \overline{\k}_{2} &= \k_{2} \widetilde{\mu}_{1,Y_2} (\widetilde{\k} ) . \label{eq:kbar}
\end{align}
By the case $m=1$, the right-hand sides of these equalities define a surjective map when restricted to the positive orthant (by omitting the zero components if some reaction rate constants are set to zero). 
In turn, $\overline{\cN}$ is an extended network of $c\ce{<=>} c'$ via the addition of the intermediates $Y_2,\dots,Y_{m}$. By the induction hypothesis, $\overline{\cN}$ satisfies the generalized realization condition. Let $\overline{\mu}_{i,c}, \overline{\mu}_{i,c'}$ for $i=2,\dots,m$ correspond to this extension.

Recall that  $\mu_{i,c}$ and $\mu_{i,c'}$ are the coefficients of $x^c$ and $x^{c'}$ respectively after writing $y_1,\dots,y_m$ in terms of $x$ by solving the steady state equations corresponding to the intermediates. This system can be solved iteratively, by first finding $y_1$ and then $y_{2},\dots,y_{m}$.
If we let 
$\varphi(\k)= \big(\k_{3} \widetilde{\mu}_{1,c} (\widetilde{\k} ),  \k_{2} \widetilde{\mu}_{1,Y_2} (\widetilde{\k} ), \k_5,\k_6,\dots,\k_{2m+2}\big)$, it follows that 
\[ \mu_{i,c}= \overline{\mu}_{i,c} (\varphi(\k)),\qquad \mu_{i,c'}= \overline{\mu}_{i,c'} (\varphi(\k)), \qquad \textrm{ for }i=2,\dots,m.\]
For $i=1$, iterative elimination of $y_1$ and   $y_2= \overline{\mu}_{2,c}(\varphi(\k)) x^{c}+ \overline{\mu}_{2,c'}(\varphi(\k)) x^{c'}$ gives that
\[ y_1=  \widetilde{\mu}_{1,Y_{2}} (\widetilde{\k} ) y_{2} + \widetilde{\mu}_{1,c}  (\widetilde{\k} ) x^{c}
 = \widetilde{\mu}_{1,Y_{2}}(\widetilde{\k} )  \overline{\mu}_{2,c'}(\varphi(\k)) x^{c'}+ \big(   \widetilde{\mu}_{1,Y_{2}} (\widetilde{\k} )  \overline{\mu}_{2,c} (\varphi(\k)) + \widetilde{\mu}_{1,c}  (\widetilde{\k} ) \big)  x^{c}.
     \]
Hence
\[ \mu_{1,c}(\k)= \widetilde{\mu}_{1,Y_{2}} (\widetilde{\k} )  \overline{\mu}_{2,c} (\varphi(\k)) + \widetilde{\mu}_{1,c}  (\widetilde{\k} ) ,\qquad 
\mu_{1,c'}(\k)=   \widetilde{\mu}_{1,Y_{2}}(\widetilde{\k} )  \overline{\mu}_{2,c'}(\varphi(\k)) . \]
Therefore 
$\phi'(\k)= \big( \k_{2m+1}\mu_{m,c}(\kappa), \k_2\mu_{1,c'}(\kappa), \sum_{i=1}^m \mu_{i,c}(\k),  \sum_{i=1}^m \mu_{i,c'}(\k) \big)$ can be written as 
\begin{align*}
\phi'(\k)&=   \big( \k_{2m+1}  \overline{\mu}_{m,c} (\varphi(\k))  , \k_2 \widetilde{\mu}_{1,Y_{2}}(\widetilde{\k} )  \overline{\mu}_{2,c'}(\varphi(\k)) , \sum\nolimits_{i=2}^{m} \overline{\mu}_{i,c}  (\varphi(\k)),  \sum\nolimits_{i=2}^{m} \overline{\mu}_{i,c'} (\varphi(\k))\big) \\& 
+ \big(0,0,  \widetilde{\mu}_{1,Y_{2}} (\widetilde{\k} )  \overline{\mu}_{2,c} (\varphi(\k)) + \widetilde{\mu}_{1,c}  (\widetilde{\k} ) , \, \widetilde{\mu}_{1,Y_{2}}(\widetilde{\k} )  \overline{\mu}_{2,c'}(\varphi(\k)) \big).
\end{align*}
Let $(k_1,k_2,\alpha_c,\alpha_{c'})\in \R^4_{>0}$, with $k_2=0$ in the irreversible case and $\alpha_{c'}=0$ if $c'$ is not an input of any intermediate. Write $\alpha_c=\alpha_{c,1}+\alpha_{c,2}$, $\alpha_{c'} = \alpha_{c',1}+\alpha_{c',2}$ such that $\alpha_{c',2}<\alpha_{c,2}$ and $\alpha_{c,1},\alpha_{c,2},\alpha_{c',1},\alpha_{c',2}>0$ ($= 0$ as appropriate). 
We want to show that $(k_1,k_2,\alpha_c,\alpha_{c'})=\phi'(\k)$ for some $\k$. 
First note that by the induction hypothesis, we can find $\overline{\k}=\big(\overline{\k}_1,\overline{\k}_2,\k_5,\dots,\k_{2m+2} )$ such that 
\[(k_1,k_2,\alpha_{c,1},\alpha_{c',1}) = \Big( \k_{2m+1}   \overline{\mu}_{m,c} (\overline{\k}), \overline{\k}_2\overline{\mu}_{2,c'} (\overline{\k}), \sum\nolimits_{i=2}^{m} \overline{\mu}_{i,c}  (\overline{\k}),  \sum\nolimits_{i=2}^{m} \overline{\mu}_{i,c'} (\overline{\k})\Big).  \]
By the last two equalities in \eqref{eq:kbar},  the decomposition of $\phi'(\k)$ above and the definition of $\varphi$, all we need is to show that there exists $\widetilde{\k} =(\k_{1}, \k_2,\k_3,\k_4)$ such that
\begin{align*}
  \overline{\k}_{1} &= \k_{3} \widetilde{\mu}_{1,c} (\widetilde{\k} ) , & \overline{\k}_{2} &= \k_{2} \widetilde{\mu}_{1,Y_2} (\widetilde{\k} ),  \\
\alpha_{c,2} &= \widetilde{\mu}_{1,Y_{2}} (\widetilde{\k} )  \overline{\mu}_{2,c} (\overline{\k}) + \widetilde{\mu}_{1,c}  (\widetilde{\k} ) ,
& \alpha_{c',2}& =  \widetilde{\mu}_{1,Y_{2}}(\widetilde{\k} )  \overline{\mu}_{2,c'}(\overline{\k}).
\end{align*}
This gives in particular that  $\overline{\k}= \varphi(\k)$.
Since $\overline{\k}$ has now been fixed, we want
\begin{align*}
  \overline{\k}_{1} &= \k_{3} \widetilde{\mu}_{1,c} (\widetilde{\k} ) , & \overline{\k}_{2} &= \k_{2} \widetilde{\mu}_{1,Y_2} (\widetilde{\k} ),   &
  \widetilde{\mu}_{1,Y_{2}}(\widetilde{\k} ) &  = \frac{ \alpha_{c',2}  }{ \overline{\mu}_{2,c'}(\overline{\k})}>0 , &    \widetilde{\mu}_{1,c}(\widetilde{\k}) &= \alpha_{c,2}- \alpha_{c',2}>0.
\end{align*}
Since the  generalized realization condition holds for $m=1$, there exist  $\k_{1},\dots,\k_{4}$ such that this system holds (or the equivalent system if some reactions are irreversible). 
This finishes the proof for the case where  there is only one block.

If there are several blocks with the same structure as \eqref{eq:reversible}, then we simply need to notice that $\phi'$ can be written as the Cartesian product of the corresponding map for each block, and $\sum_{i=1}^m \mu_{i,c}$ can be split as a sum of the $\mu_{i,c}$'s of each block. Since the generalized realization condition holds for each block, it also holds for the whole network by splitting $\alpha_c$ accordingly for each complex $c$.

		\medskip
		\noindent
		(iii) The core network has $p$ reactions $c_0\rightarrow c_1,\dots,c_0\rightarrow c_p$.  We have
		$\phi_{c_0\rightarrow c_i}(\k)=\ell_i\mu_{m,c_0}$.   The denominator of $\mu_{m,c_0}$ is a multiple of $\sum_{i=1}^p\ell_i$ and $\mu_0=(\sum_{i=1}^p\ell_i)\mu_{m,c_0}$ does not depend on any $\ell_i$. Note that the scalar-valued function $\mu_0$ is positive and linear in   $(\k_{c_0\rightarrow Y_1},\dots,\k_{c_0\rightarrow Y_m})$. Hence by varying the reaction rate constants different from $\ell_i$, 
$\mu_0$ covers $\R_{>0}$.
With this we have that given $k_1,\dots,k_p>0$, we define $\ell_i= k_i$ and choose the rest of reaction rate constants such that $\mu_0=  \sum_{i=1}^p k_i$.
Then   $\phi_{c_0\rightarrow c_i}(\k)=k_i$, showing that \blue{$\phi^*$} is surjective. 
	\end{proof}

\smallskip
\paragraph{Acknowledgements. }  This work has been supported by the Independent Research Fund of Denmark.  We thank Alicia Dickenstein, Martin Helmer and Ang\'elica Torres for comments on a preliminary version of this manuscript.

\end{document}